\documentclass{article}
\usepackage[utf8]{inputenc}
\usepackage[margin=2.5cm]{geometry}
\usepackage{graphicx}
\usepackage{amsmath,amssymb}
\usepackage{booktabs}
\usepackage{natbib}
\usepackage{color}
\usepackage{soul}
\usepackage{verbatim}
\usepackage{bm}
\usepackage{lineno}
\usepackage{hyperref}
\usepackage[affil-it]{authblk}

\usepackage[margin=2.5cm]{geometry}
\usepackage{titlesec}
\usepackage{textcomp}
\usepackage{chngcntr}
\usepackage{parskip}
\usepackage{wrapfig}
\usepackage{geometry}
\usepackage{multicol}
\usepackage{adjustbox}
\usepackage{float}
\usepackage[utf8]{inputenc}
\usepackage[T1]{fontenc} 					
\usepackage{graphicx}
\usepackage{caption}
\usepackage{subcaption}
\usepackage{graphicx}       					
\usepackage{amsmath,amssymb} 
\usepackage{xfrac}

\usepackage{amsthm}
\theoremstyle{definition}

\theoremstyle{theorem}

\newtheorem{thm}{Theorem}

\newcommand{\R}{\mathbb{R}}

\def\N{{\rm N}}

\def\lnorm{{\rm LogNormal}}

\def\fix{{\rm fix}}
\def\ibs{{\rm IBS}}
\def\bs{{\rm BS}}
\def\lp{{\rm LP}}
\def\bp{{\rm BP}}
\def\GP{{\rm GP}}
\def\GPt{{\rm GP-t}}
\def\hb{{\rm HB}}

\title{Combining predictive distributions for time-to-event outcomes \\ in meteorology}
\author[1,2*]{C\' eline Cunen}
\author[2]{Thea Roksv\aa g}
\author[2]{Claudio Heinrich‐Mertsching}
\author[2]{Alex Lenkoski}

\affil[1]{Oslo Centre for Biostatistics and Epidemiology, Oslo University Hospital, P.O.Box 1122 Blindern
0317 Oslo, Norway}
\affil[2]{Norwegian Computing Center, P.O.Box 114 Blindern
0314 Oslo, Norway. (roksvag@nr.no; claudio.heinrich@hotmail.de; lenkoski@nr.no)}
\affil[*]{Corresponding author: C\' eline Cunen, cmlcunen@uio.no}
\date{ }

\begin{document}

\maketitle

\begin{abstract}
\noindent
Combining forecasts from multiple numerical weather prediction (NWP) models have shown substantial benefit over the use of individual forecast products.  
Although combination, in a broad sense, is widely used in meteorological forecasting, systematic studies of combination methodology in meteorology are scarce.
In this article, we study several combination methods, both state-of-the-art and of our own making, with a particular emphasis on situations where one seeks to predict \textit{when} a particular event of interest will occur. Such time-to-event forecasts require particular methodology and care. 
We conduct a careful comparison of the different combination methods through an extensive simulation study, where we investigate the conditions under which the combined forecast will outperform the individual forecasting products.
Further, we investigate the performance of the methods in a case-study modelling the time to first hard freeze in Norway and parts of Fennoscandia.
\end{abstract}



\section{Introduction}

 In many real forecasting situations one may have access to multiple forecasts based on different sources of information. In meteorology, for example, combining predictions from several numerical weather prediction (NWP) models frequently results in better predictions than any of the single models. In this article, we investigate how to best combine probabilistic weather forecasts, with a specific focus on time-to-event prediction. We also attempt to clarify, through simulations and a case-study, under what condition combination will be fruitful in the above sense, that the combined forecast is more skilful than any of the single-source forecasts.

Combination of weather forecasts from different sources has been studied be several authors, see for instance \citep{vislocky1995,kober2012,gneiting2013,baran2018,schaumann2020}. The terms combination, aggregation, merging and blending are used more or less interchangeably. Older works combine the forecasts by taking simple linear averages of the different forecast distributions \citep{vislocky1995, kober2012}, often giving equal weights to the different forecasts. Later contributions have explored more complex combination methods and typically let combination weights, and sometimes other parameters, be estimated from historical data \citep{gneiting2013, baran2018, schaumann2020}. In the field of econometrics, many similar methods have been developed and studied in parallel, see \citet{wang2022} for a recent review. 

Meteorological forecasts are different from forecasts in most other fields, both in terms of their origin and in terms of the form they are issued in. Unlike fields where forecasts are the result of statistical models on historical data, weather forecasts are produced from NWP models, i.e. physical models of the atmosphere and oceans \citep{baran2018}.

Further, meteorological forecasts usually come in the form of ensembles, where each ensemble member corresponds to a run of the model with different initial conditions and model physics \citep{baran2018}. The different ensemble members are typically considered as exchangeable in the statistical sense, and the ensemble may be interpreted as a random sample from the underlying forecasting distribution. In this article, ensembles from different types of NWP models constitute the different forecasts which we want to combine. 
This could for example be ensembles from different weather centres, 
forecasts with different release frequencies and lead times, or a mix of dynamic and statistical forecasts.
Our study will focus on the combination of seasonal and subseasonal forecasts.

Seasonal forecasts provide forecasts of weather conditions up to 5-7 months into the future, while subseasonal forecasts provide weather information for the next 2-3 months. Recent studies from Europe show that subseasonal and seasonal forecasts often have skill the first 3-4 weeks, particularly for temperature-based variables \citep{hyvar2020,ceglar2021,baker2023,roksvaag2023}, and occasionally skill for even longer lead times \citep{hyvar2020_b}. 

Many users want to use meteorological forecasts to anticipate \textit{when} a particular event of interest will occur. For example drought prediction \citep{yuan2013}, prediction of the onset of the next rainy season \citep{scheuerer2024}, predictions of the time-to-hard-freeze for agricultural application \citep{roksvaag2023} and prediction of the flowering time for wheat \citep{ceglar2021}. Time-to-event outcomes such as these should be handled with methods from \textit{survival analysis}. This subfield has deep roots in medical statistics, but is also frequently used in a diverse set of fields ranging from engineering to economics.  To our knowledge, \citet{roksvaag2023} is the first instance of the use of survival time methods in meteorology.  \citet{roksvaag2023} employed survival time methods based on the output of a single NWP model.  Our present work extends this methodology to consider how forecasts from multiple sources can be combined effectively.

Despite the relatively large literature on the topic of forecast combination, there exists only a few papers which have systematically studied and compared different combination methods in the context of meteorological forecasting. \citet{baran2018} is an exception. They compare several forecast combination methods, but are specifically focusing on combining two statistical distributions fitted to the \textit{same} forecast ensemble in the context of post-processing. 
In their conclusion they highlight that the combination methods they have studied may also be used to combine \textit{different} forecasts, and this is what we have addressed in the present paper. Through simulation experiments and a real data case-study, we  provide a thorough study of combination methods, particularly relevant for the combination of forecasts in meteorology. Our case-study and simulation experiment are time-to-event set-ups, but the conclusions we draw about the combination methods should generalise to other outcomes in meteorology. 

Our contribution in this paper is (1) a careful comparison of different combination methods, through an extensive simulation study and a real, interesting application; (2) two new combination methods; (3) application of (more) survival analysis tools to the weather prediction context. In the next section we present the real data application which will serve as an illustration and case-study throughout the paper. Then in Section~\ref{sec:methods} we describe the methods that will be used: the necessary survival analysis methodology, as well as the combination methods, and the framework for forecast evaluation. In Section~\ref{sec:illustration} we return to the application to demonstrate the use of the combination methodology in a small example. Section~\ref{sec:simulations} contains the simulation study, and Section~\ref{sec:application} the case-study. Finally we discuss our findings and conclude in Section~\ref{sec:conclusion}.

\section{Motivating the case-study}
\label{sec:motivation}

Throughout the paper we will consider a particular forecasting application, which we describe here and return to in Sections~\ref{sec:illustration} and \ref{sec:application}. The application also serves as an inspiration for the design of the simulation study in Section~\ref{sec:simulations}. Our goal in the case-study is to predict the time to the first hard freeze of the season, subsequently referred to as the time-to-hard-freeze.  Hard freeze is defined as a mean daily near-surface air temperature below 0 $^\circ$C. Such definitions should be made with the forecast user in mind. The definition of hard freeze was for example here coordinated with Norwegian farmers, as our study locations are in Norway and parts of Fennoscandia. In this region the first hard-freeze of the season typically comes some time after October 1.

Two ensemble forecast products for surface temperature with forecasts covering the relevant time period are available,
\begin{enumerate}
    \item the seasonal forecast: an ensemble of around 50 members with issue date September 1 and a maximal lead time of 122 days (December 31); 
    \item the subseasonal forecast: an ensemble of 11 members with issue date September 30 and a maximal lead time of 46 days (November 15).
\end{enumerate}
The two forecasting products will be referred to as our two sources of information in the following. 
The seasonal forecasts we use here were also considered in \citet{roksvaag2023}, where the authors demonstrate that seasonal forecasts alongside survival analysis methods can produce skilful forecasts of the time-to-hard-freeze. The question is now, can we obtain  an even more skilful forecast by combining the seasonal and subseasonal forecasts? 

\begin{figure}
\includegraphics[width =1.05\textwidth]{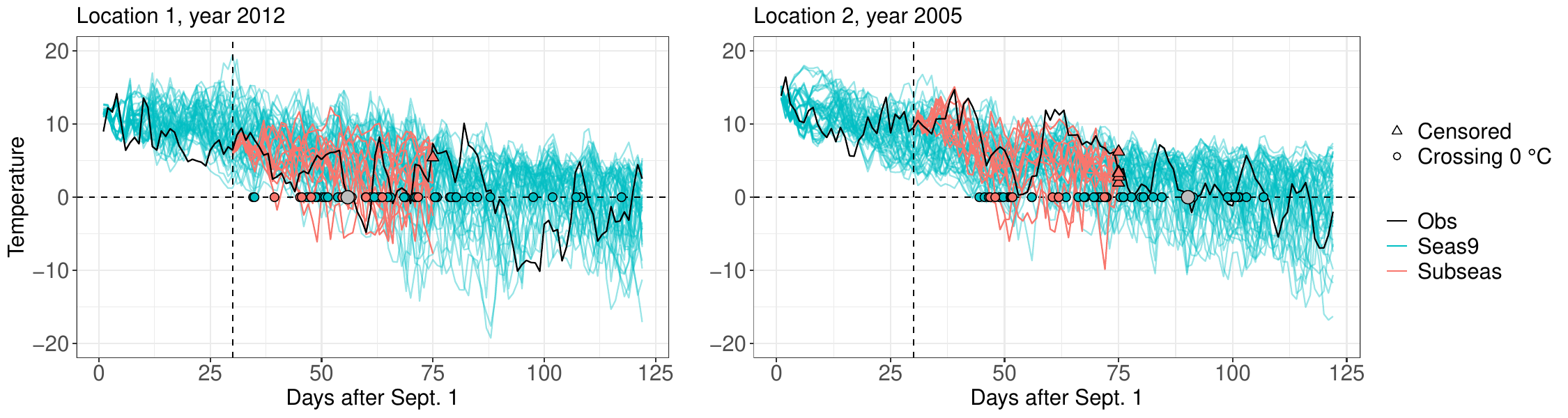}
\caption{Temperature trajectories for ensemble members of the seasonal forecast (Seas9), issued September 1, and the subseasonal forecast (Subseas), issued September 30 for two different locations in Fennoscandia for two different years. The observed temperature is shown in black (Obs). The first time each temperature trajectory crosses $0 ^\circ$ C is marked with a circle. Members that never cross zero before the maximum lead time is reached are marked with a triangle and are treated as censored in the survival analysis framework. \label{fig:motivation}}
 \end{figure}
 
Figure~\ref{fig:motivation} provides the observed temperature and forecasted temperature trajectories over the autumn season in two different years (and locations). Each coloured line represents the forecasted temperature trajectory from a particular ensemble member. 

From each forecasted temperature trajectory we obtain a forecast for the time-to-hard-freeze by recording the time when the trajectory first crosses 0 $^\circ$C. This results in an ensemble of forecasted times-to-hard-freeze from each of the two information sources. These time-to-event forecasts are shown as blue (seasonal) or red (subseasonal) dots in the figure. The actual observed time-to-hard-freeze is given as a grey dot. The figure illustrates that some trajectories will reach their maximal lead time without crossing the 0 $^\circ$C line. In the figure these instances are indicated by red triangles, and these time-to-event observations are considered to be \textit{censored}. The exact value of a censored time-to-event is not known; we only know that the event of interest will happen at some point after the censoring date. Censoring is one of the central challenges in survival analysis, as it is specific to this branch of probabilistic modelling.

In the two examples shown in Figure~\ref{fig:motivation} we only have censored observations in the subseasonal ensemble. This will often be the case in our case-study as well, due to the short lead time of the subseasonal forecast. Note however that we can have censoring in the seasonal ensemble too, and for some particular combinations of year and location we can even have censored realisations, meaning that hard freeze does not occur before the 31st of December (and the black line in the plots above does not cross below 0  $^\circ$C).

In our case-study we will assume that we are on the 1st of October and want to predict the time-to-hard-freeze, if hard freeze has not happened already. This date is indicated by a vertical, dashed line in Figure~\ref{fig:motivation}. The date was chosen out of convenience, and hard freeze happens after October 1st for most of the 103 study locations we treat in the case-study (see Section~\ref{sec:application}). On October 1st the seasonal forecast is already old (having been issued on September 1), but as we will see later, it still has some predictive skill, at least when combined with the subseasonal forecast.

The data sources for the observations and forecasts are described in detail in Appendix~\ref{app:DataSources}.

\section{Methods}
\label{sec:methods}

We will denote the variable we want to forecast as $T$, indicating a time-to-event, for example the time-to-hard-freeze in our motivating application.  
The combination methods we study require estimation in two steps. In the first step, each source is analysed separately to produce source-specific forecast distributions. In the second step, these source-specific forecast distributions are combined. This step typically requires the estimation of combination weights and potentially other combination parameters.
In the following section, we start by presenting methods we will use in the first step. These are standard methods from the field of survival analysis. Then, in Section~\ref{sec:combometh} we present the combination methods, most of which have utility in contexts beyond time-to-event studies.

\subsection{Single-source analyses}
\label{sec:survmeth}

Our information sources are meteorological forecasting products which come in the form of \textit{ensembles}. So for a particular forecasting situation (i.e. a particular time and place) and a specific information source $k$, we have access of to a sample of $n_k$ time-to-event forecasts which we will use to estimate a source-specific forecasting distribution $\widehat F_k(t)$, in the form of a cumulative distribution function (CDF). The hat highlights that $F_k(t)$ is estimated from the ensemble members, but we will sometimes omit the hats in the following for ease of notation. In survival analysis, it is typical to focus on the survival function $S_k(t)$ instead of the CDF, i.e. $S_k(t) = 1- F_k(t)$, and we will adhere to that convention in the following. The survival function gives the probability that the event of interest, eg. hard freeze, has not occurred before time $t$.

The survival function $S_k(t)$ can be estimated using either non-parametric or parametric tools. We will present some of these time-to-event tools here. 
The `observations' in the ensemble will be of the typical survival data form $(t_1, \delta_1),\dots,(t_{n_k}, \delta_{n_k})$ where $t_i$ is the possibly censored time-to-event and $\delta_i$ is the censoring indicator, with $\delta_i=1$ indicating that the event occurred at time $t_i$ and $\delta_i=0$ indicating that the observation was censored at time $t_i$.

The Kaplan-Meier estimator \citep{kaplan1958} is the standard non-parametric estimator of the survival function, 
\begin{align} \label{eq:KM}
\widehat S_k(t) = \prod_{t_i\leq t}\left(1 - \frac{d_{i}}{n_{i}}\right)
\end{align}
where \( t_i \) are the ordered event times, \( d_i \) is the number of events at time \( t_i \), and \( n_i \) is the number of individuals at risk just before time \( t_i \). This estimator provides a step-function estimate of the survival probability over time.

Alternatively, one may use parametric methods to estimate $S_k(t)$. In this paper we will only consider log-normal models, but there are of course many other options. The log-normal is typically parameterised by a mean $\xi$ and a variance parameter $\tau^2$ on the log-transformed scale. The CDF of the log-normal may be written as 
\begin{align} \label{eq:lognormal}
    F_k(t) = \Phi \left( \frac{\log(t) - \xi}{\tau} \right),
\end{align}
where  $\Phi(\cdot)$ is the standard normal CDF. We will estimate the source-specific parameters $(\xi,\tau)$ by either maximum likelihood or by minimizing the integrated Brier score (see Eq.~\eqref{eq:ibs}). In either case, we obtain a continuous $\widehat S_k(t)$ by plugging-in the estimated parameters into Eq.~\eqref{eq:lognormal} above. See Appendix~\ref{app:Snew} for a slightly more advanced method to get $\widehat S_k(t)$ from the estimated log-normal model. 

In a time-to-event set-up the log-likelihood takes the following form,
\begin{align} \label{eq:survML}
    \ell_{n_k}(\xi,\tau) = \sum_{\delta_i=0} \log S_k(t; \xi, \tau) +  \sum_{\delta_i=1} \log f_k(t; \xi, \tau),
\end{align}
where the first part consists of a sum over the censored observations, and the second part over the non-censored observations. In our case $S(t; \xi, \tau)$ will be the log-normal survival function and $f(t; \xi, \tau)$ the log-normal density. This allows us to estimate $(\xi,\tau)$ taking into account the censored observations.

\subsection{Combination methods}
\label{sec:combometh}

Here we present methods for the second step of the combination procedure, where the source-specific forecasting distributions are combined to obtain a combined forecasting distribution, $S_c(t)$. We present 
four combination methods, which we denote by LP, BP, GP and HB (acronyms to be explained below). The first two, LP and BP, are state-of-the-art methods in the literature, while the two next ones are our own contributions. 
We will present the methods for the case where one has two sources of information, but it is straightforward to generalise to more than two sources. From the source-specific analyses in the first step (Section~\ref{sec:survmeth}) we obtain two estimated survival functions  $\widehat S_1(t)$ and $\widehat S_2(t)$, which may be either step-functions or continuous as we have seen above. From these we easily have the estimated CDFs $\widehat F_1(t)$ and $\widehat F_2(t)$ which we will use in the combination step. Note that we omit the hats on the single-source CDFs below, for ease of notation. 

\subsubsection{Linear pooling (LP)}
Linear pooling combines $F_1(t)$ and $F_2(t)$ by a taking a simple weighted sum,
\begin{align}\label{eq:LPmethod}
    S^{\lp}_c(t) = 1-[ \omega F_1 (t) + (1-\omega)F_2(t)].
\end{align}
The combined survival curve $S^{\lp}_c(t)$ depends on a weight $\omega$ which reflects the quality of source 1 relative to source 2. The parameter $\omega$ is usually estimated from training data, i.e. past forecast-observation pairs, and this estimation step will be explained in Section~\ref{sec:estimation}. Sometimes it is instead simply set to $\omega=0.5$.

\citet{gneiting2013} studied this combination method in some detail, and proved that if the two single-source forecasts are calibrated then the combined forecast $S_c(t)$ will be overdispersed.  Thus $S_c(t)$ will not be calibrated, and this is a consequence of the LP failing to be \textit{flexibly dispersive} in the terminology of \citet{gneiting2013}.

The linear pool method is in widespread use within econometrics, see \citet{aastveit2018} and \citet{hall2007}, and has also been used in meteorological contexts, see for instance \citet{baran2018} and \citet{schaumann2020}.

\subsubsection{Beta pooling (BP)}

\citet{gneiting2013} proposed the beta-transformed linear pool, which we will simply call the beta pooling method, or BP. 
Here the single-source CDFs are combined by the following formula,
\begin{align} \label{eq:BP}
   S^{\bp}_c(t) = 1-B_{ \alpha,  \beta}[ \omega F_1 (t) + (1- \omega) F_2(t) ]. 
  \end{align}
where $B_{\alpha,\beta}(\cdot)$ is the beta CDF with parameters $\alpha$ and $\beta$. Note that if $\alpha=\beta=1$ we have $S^{\bp}_c(t) = S^{\lp}_c(t)$, because in that case the beta distribution is equal to a uniform.
The combination parameters $(\alpha, \beta, \omega)$ have to be estimated from training data, see Section~\ref{sec:estimation}.
\citet{gneiting2013} prove that the BP is flexibly dispersive, and it will therefore produce a combined forecast $S^{\bp}_c(t)$ which is calibrated, even when both single-source forecasts also are calibrated.

The BP is in less widespread use than LP, but is for example investigated in \citet{baran2018}, and extended in \citet{bassetti2018}.

\subsubsection{Gaussian pooling (GP)}

The Gaussian pooling, or GP, combines the single-source CDFs in the following way,
\begin{align} \label{eq:GP}
   S^{\GP}_c(t) = 1-\Phi\left [\frac{ \omega\Phi^{-1}\{F_1 (t)\} + (1- \omega) \Phi^{-1}\{F_2(t)\} - \mu}{ \sigma} \right]. 
\end{align}
where $\omega$ is the weight given to source 1, and $\Phi(\cdot)$ and $\Phi^{-1}(\cdot)$ are the standard normal CDF and quantile function respectively. Here we have three combination parameters $(\mu, \sigma, \omega)$, which again must be estimated from training data. The role of the combination parameter will be investigated in the illustration in Section~\ref{sec:illustration} and in the simulation study.

 A version of this method, without the parameters $\mu$ and $\sigma$, has been considered in other parts of the literature, for example very briefly in \citet{gneiting2013} under the umbrella of generalised linear combination formulas.  
 Similarly to the BP method, the three-parameter GP is flexibly dispersive, see Theorem~\ref{th1} in the Appendix.
We find that the combination parameters in the GP have a more straightforward interpretation, than the combination parameters in BP.

Further, the GP has a very natural extension to situations where there the training dataset is small. 
Then, it is natural to use the following student-t version of GP, which we will call GP-t
 \begin{align} \label{eq:GPt}
   S^{\GPt}_c(t) = 1-G_{n-1}\left [\frac{ \omega\Phi^{-1}\{F_1 (t)\} + (1- \omega) \Phi^{-1}\{F_2(t)\} -  \mu}{ \sigma} \right],
\end{align}
 where $G_{n-1}(\cdot)$ is the CDF of the student-t distribution with $n-1$ degrees of freedom, and $n$ is the number of training samples.

 \subsubsection{Hazard blending (HB)}

The previous three combination methods combine the \textit{CDFs} from the two sources, while hazard blending, HB, works on the \textit{hazard} scale,
  \begin{align*}
   S^{\hb}_c(t) = \prod_{t_i \leq t} (1-\lambda_i), \text{ with } \lambda_i=\frac{ \omega d_{1,i} + (1-  \omega) d_{2,i}}{\omega n_{1,i} + (1-  \omega) n_{2,i}}.
  \end{align*}
Here $d_{1,i}$ is the number of ensemble members from source 1 for which the event of interest occurred at time $t_i$, while $n_{1,i}$ is the number of ensemble members from source 1 for which the event has not yet occurred at time $t_i$. The numbers 
$d_{2,i}$ and $n_{2,i}$ have similar roles, for source 2.  

This combination method is inspired by the long-established Kaplan-Meier estimator, see Eq.~\eqref{eq:KM}, which is the standard non-parametric estimator of a survival function in the presence of censored observations.  HB is thus a non-parametric combination method, which does not require any parametric model for $F_1$ and $F_2$. HB will always result in a survival function $S^{\hb}_c(t)$ in the form of a step-function, while the three other combination methods may result in either step-functions or continuous functions depending on the analyses in the first step of the combination procedure (Section~\ref{sec:survmeth}).

We will refer to hazard blending (HB) and linear pooling (LP) as \textit{simple} combination methods, in the sense that they only have a single combination parameter $\omega$. The $\omega$ in HB and LP cannot be directly compared, however, since they belong to different scales. 
In contrast to LP and HB, we will refer to beta pooling (BP) and Gaussian pooling (GP) as \textit{complex} combination methods, since they have three combination parameters each.

\subsection{Forecast evaluation}
\label{sec:eval}


In our simulation study and in the real data application, we will evaluate the forecasts, either from single-sources or from combinations, with respect to sharpness under calibration, as is commonly done \citep{gneiting2007}. 
Calibration will be assessed using the probability integral transform (PIT) histograms (see for instance \citet{thorarinsdottir2018}), and also comparing the mean and standard deviation of the PIT values with 0.5 and $\sfrac{1}{\sqrt{12}}$ respectively, which should be their values under perfect calibration.

Predictive performance will be assessed by the integrated Brier score (IBS). If we have $m$ realised times-to-event on which to evaluate our forecasts, the Brier score \citep{brier1950} takes the form
\begin{align*}
    \bs(t) = \frac{1}{m}\sum_{i=1}^m [ \mathbb{I}\{T_i > t \} - S_{i}(t) ]^2,
\end{align*}
where $S_i(t)$ is the forecast which we want to evaluate and $T_i$ is the observed time-to-hard-freeze. The Brier score is a function of time $t$ so in order to obtain a single numerical value we `integrate' the score over time $t$. In practice, we will follow \citet{roksvaag2023} and discretise the timespan into days and sum from one to $t_{\max}$, the maximum number of days into the future we consider,
\begin{align} \label{eq:ibs}
    \ibs =  \sum_{t=1}^{t_{\max}} \bs(t) =\sum_{t=1}^{t_{\max}} \frac{1}{m}\sum_{i=1}^m [\mathbb{I}\{t_i > t \} - S_{i}(t)]^2.
\end{align}
In the case-study we use $t_{\max}=120$. The IBS is equivalent to the continuous ranked probability score (CRPS). 

\subsection{Estimation of combination parameters}
\label{sec:estimation}
The combination methods contain combination parameters, the weight $\omega$ and potentially others, which are typically estimated from a number of available forecast-observation pairs. For ease of exposition, we assume here that we have these training data from a number of years at a \textit{single} location. The training dataset then consists of (1) $n$ past time-to-event realisations  at the particular location $(t_i: i=1,\dots,n)$, and (2) past ensemble forecasts from each of the two sources for each year $i$. 
We will consider two estimation methods, one based on maximising the likelihood and the other on minimizing the IBS.

Note that in our case-study and also in other time-to-event applications we may have censored \textit{realisations}, in addition to the censored ensemble members which were described in Section~\ref{sec:survmeth}. These censored realisations need to be accounted for in the estimation of the combination parameters, specifically the log-likelihoods below need to be amended in a similar way as the log-likelihood in Eq.~\eqref{eq:survML}.

\subsubsection{Maximum likelihood (ML)}
\label{sec:combML}

The parameters of the combination methods, LP, BP and GP, can be estimated by maximising the likelihood of the combination formula over the training data. This may for example look like the following for LP,
\begin{align*}
    \ell(\omega) = \sum_{i=1}^n \log[ \omega f_{1,i}(t_i) + (1-\omega) f_{2,i}(t_i) ],
\end{align*}
where $f_{1,i}(\cdot)$ and $f_{2,i}(\cdot)$ are the forecast densities for year $i$ from the two sources. 

For the BP method the log-likelihood can be found  in \citet{gneiting2013} or by differentiating Eq.~\eqref{eq:BP}.
For the GP method differentiating Eq.~\eqref{eq:GP} results in the following log-likelihood
\begin{align*}
    \ell(\omega,\mu,\sigma) = 
    \sum_{i=1}^n \log \{ \phi\left( \frac{\omega\Phi^{-1}\{F_{1,i} (t_i)\} + (1-\omega) \Phi^{-1}\{F_{2,i}(t_i)\}- \mu}{\sigma} \right) 
    \frac{1}{\sigma} &[  \omega\Phi^{-1}{}'\{F_{1,i} (t_i)\}f_{1,i}(t_i) + \\ &(1-\omega) \Phi^{-1}{}'\{F_{2,i}(t_i)\}f_{2,i}(t_i) ]  \},
\end{align*}
where $\phi(\cdot)$ is the standard normal density function, $\Phi^{-1}{}'(\cdot)$ denotes the derivative of the standard normal quantile function (which may not have an analytical expression, but is simple to evaluate numerically). $F_{1,i}(\cdot)$ and $F_{2,i}(\cdot)$ are the forecast CDFs and $f_{1,i}(\cdot)$ and $f_{2,i}(\cdot)$ the forecast densities for year $i$ for each of the two sources. 

In the context of combination, this estimation method is used in for example \citet{gneiting2013} and in \citet{baran2018}.

\subsubsection{Minimising IBS}
\label{sec:minIBS}

Instead of ML, one may find the combination parameters by numerically optimising other loss functions, for instance the IBS. For example, for the HB method we would minimise
\begin{align*}
    \ibs(\omega) = \sum_t \frac{1}{n}\sum_{i=1}^n (\mathbb{I}\{t_i > t \} - S^{\hb}_{c,i}(t))^2,
\end{align*}
with respect to $\omega$. Here $S^{\hb}_{c,i}(t)$ is the predictive survival curve from year $i$ and $t_i$ is the realised time-to-hard-freeze in year $i$. This method may be used for LP, GP and BP too. This estimation method was used in the context of combination in for example \citet{baran2018}.

\section{Illustrating combination}
 \label{sec:illustration}

For concreteness we will show the combination methods at work in a specific example. Assume we are in the location of Torsby in Sweden on the 1st of October 2018 and we want to predict the time-to-hard-freeze. As explained in Section~\ref{sec:motivation} we have two forecasting products which will be our two sources of information: the seasonal and subseasonal forecasts.

This is the same data which we will use in the case-study in Section~\ref{sec:application}, but here we only look at one specific location.

We have 19 forecast-observation pairs for the years between 1999 and 2017 -- these will be our training data, from which we will estimate combination parameters.
In Figure~\ref{fig:ex1} we see the forecasts from the two sources for some of the training years. In some years the two sources are similar, other years they differ substantially. Sometimes the seasonal forecast is better (for example in 2010), while for other years the sub-seasonal is closest to the realised time-to-hard-freeze (for example in 2011). We see that some of the ensemble members in the subseasonal forecast are censored in 2010: 3 of the 11 members predict that frost will come after day 46, i.e. some time after mid November, beyond the lead time of the subseasonal forecast. 

 \begin{figure}[h] 
\begin{center}
\includegraphics[scale=0.55]{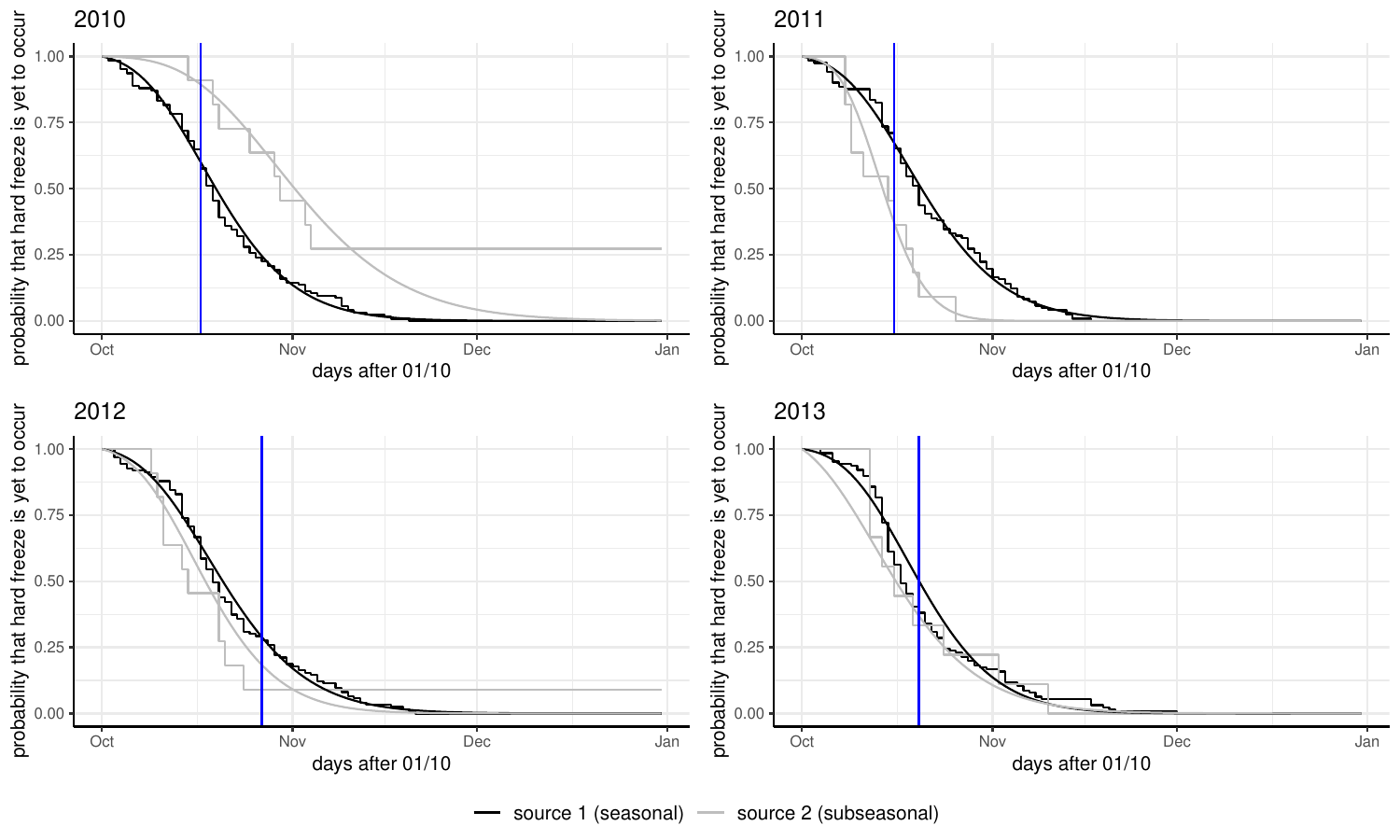}

\end{center}
\caption{\small Four years of the training data (past forecast-observation pairs). The vertical blue line marks the realised time-to-hard-freeze in each year. The grey and black step-functions are the KM estimators for the subseasonal forecast (grey) and seasonal forecast (black). The smooth functions in the same colours are log-normal models fitted to these ensembles.
}
\label{fig:ex1}
\end{figure}

In the first step of the combination procedure, we fit parametric log-normal models to the single-source forecasts in each year, as described in Section~\ref{sec:survmeth}. The resulting continuous survival curves will be used in the next step of the combination procedure by some of the combination methods (LP, BP and GP in this case). These parametric survival curves are shown as smooth lines in Figure~\ref{fig:ex1}, alongside the step-functions which are the KM estimators for each of the two sources. 
We see that the log-normal models fit the ensemble data reasonably well. 

In the second step of the combination methods, we estimate the combination parameters based on the 19 years of training data. Here we have used maximum likelihood for LP, BP, and GP and minIBS for HB. We obtain the following estimated combination parameters:
\begin{align*}
    &\widehat \omega_{{\text LP}} = 0.49 \quad \text{(weight on seasonal)}\\
    &\widehat \omega_{{\text BP}} = 0.49 \quad \hat \alpha= 1.63 \quad \hat \beta = 1.21\\
    & \widehat \omega_{{\text GP}} = 0.79 \quad \hat \mu= 0.22 \quad \hat \sigma =0.92 \\
    &\widehat \omega_{{\text HB}} = 0.09 
\end{align*}
Note again that $\omega$ for LP and HB are on different scales. The one from LP is more directly interpretable and indicates that the two sources are considered to be approximately equally good, on average, for this location in the years 1999 to 2017. The additional combination parameter in the GP method, $\mu$ and $\sigma$, can correct potential biases and over- or under-dispersion in the combined forecast. Here $\mu$ is relatively close to 0, so there does not seem to be much bias in the combined forecast. The standard deviation $\sigma$ is a bit below 1, which is a typical value when both single-source forecasts are close to calibrated. 
The parameters $\alpha$ and $\beta$ in BP play a similar role as $\mu$ and $\sigma$, but are a bit less directly interpretable. The further they are from (1,1) the more they correct the mean and variance of the combined forecast.

 \begin{figure}[h] 
\begin{center}
\includegraphics[scale=0.45]{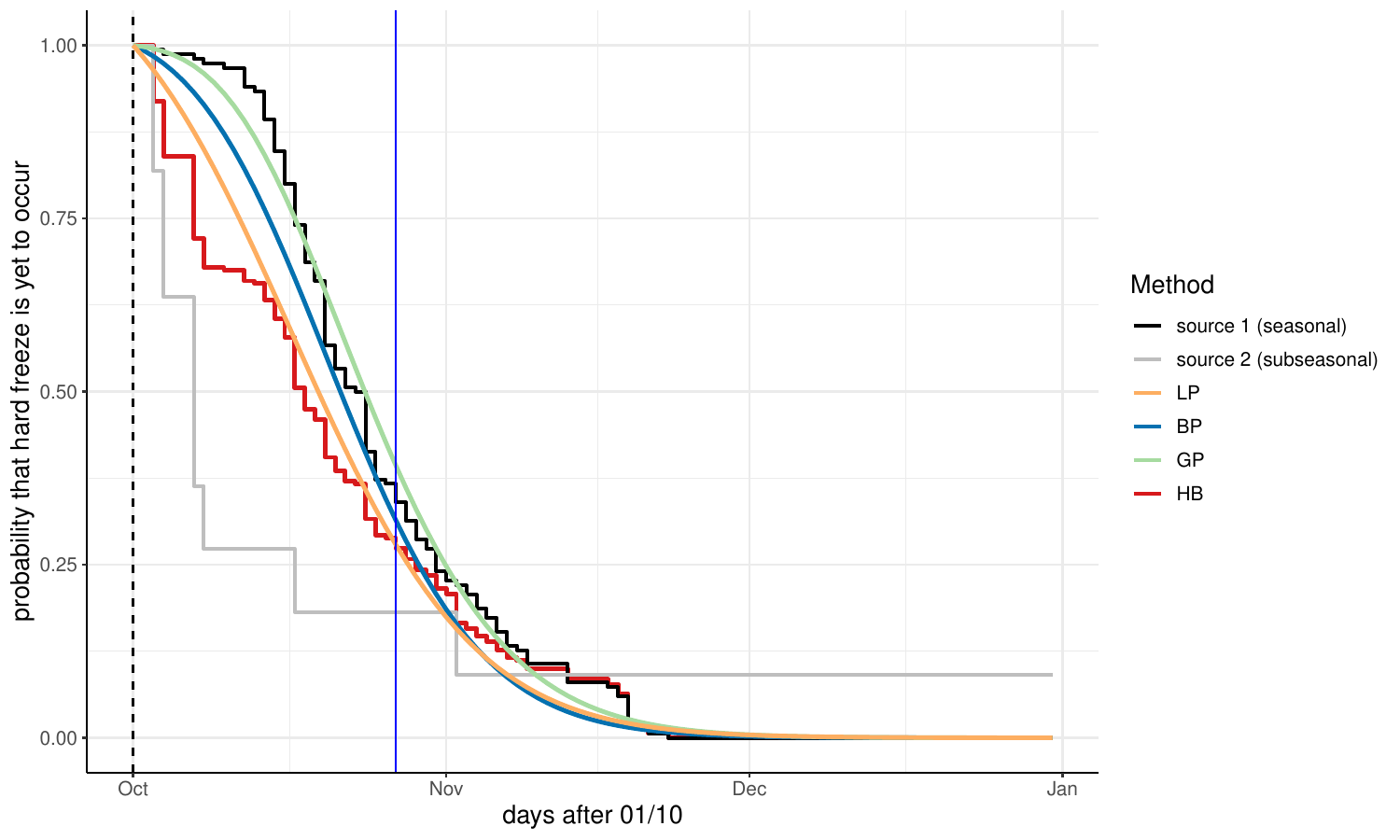}
\end{center}
\caption{\small Forecasts for 2018. The grey and black step-functions are the KM estimators for the subseasonal forecast (grey) and seasonal forecast (black). The coloured survival curves are the resulting predictive survival curves from the four combination methods. The vertical blue line marks the realised time-to-hard-freeze.
}
\label{fig:ex2}
\end{figure}

In Figure~\ref{fig:ex2} we see that the predictive survival curves from the combination methods tend to lie between the survival curves of source 1 and source 2. This will typically be the case, but for BP and GP the combined survival curve can lie outside the two single-source survival curves -- in cases where there is bias in one or both sources.  HB produces a step function. The others produce smooth curves since we used the parametric single-source survival curves in this case. LP and HB are otherwise quite similar and that is typical in our experience.
Compared to the other methods, GP is closer to the seasonal here -- as we also see from the estimated weight $\omega_\GP$. The realised time-to-hard-freeze in 2018 was in late October.

\section{Simulation study}
\label{sec:simulations}

Here we study the four combination methods in controlled conditions. We strive to mimic our application so we generate time-to-hard-freeze observations in a single location, over a number of years indexed by $i$. In each year we have access to two different forecasts from two sources of information: source 1 and 2. The forecasts are issued as ensembles, with $n_1$ samples from source 1 and $n_2$ samples from source 2, from which we estimate single-source survival curves $S_{1,i}(t)$ and $S_{2,i}(t)$ for each year $i$. Then, we estimate combination parameters across many years.

\begin{table}[ht]
\caption{\small Probabilistic forecasts compared in the simulation experiment. Source-specific estimation can be either parametric (ML) or non-parametric. Combination estimation can be by ML or by minimising the IBS. The last column indicates the number of combination parameters. }
\centering \small
\begin{tabular}{cccc} \toprule
    Method & Source-specific est. & Combination est. & No. combo. params. \\ \midrule 
  \textbf{source 1} & \textbf{ML} (log-normal) &  - &  -  \\ 
  \textbf{source 2} & \textbf{ML} (log-normal) &  - &  -  \\ 
  source 1 & KM (non-parametric) &  -  &  - \\ 
  source 2 & KM (non-parametric) &  - &  -  \\ 
  \textbf{LP} & \textbf{ML} (log-normal) & \textbf{ML} &  \textbf{1} \\
  \textbf{BP$_3$} & \textbf{ML} (log-normal) & \textbf{ML} &  \textbf{3} \\
  \textbf{GP$_3$} & \textbf{ML} (log-normal) & \textbf{ML} &  \textbf{3}\\
  \textbf{HB} & \textbf{non-parametric} & \textbf{min IBS} &  \textbf{1} \\
  LP$_{\ibs}$ & ML (log-normal) & min IBS &  1 \\
  BP$_{\ibs}$ & ML (log-normal) & min IBS &  3\\ 
  GP$_{\ibs}$ & ML (log-normal) & min IBS &  3\\
  BP$_2$ & ML (log-normal) & ML with $\alpha=\beta$ &  2 \\
  GP$_1$ & ML (log-normal) & ML with $\mu=0$ and $\sigma=1$ &  1 \\
  GP$_2$ & ML (log-normal) & ML with $\mu=0$ &  2\\
  \textbf{GP$_3$-t} & \textbf{ML} (log-normal) & \textbf{ML} &  \textbf{3}\\ \midrule
  LP$_0$ & - & - &  0\\
  merge & ML (log-normal) & - &  0\\ \bottomrule
\end{tabular}
\label{tab:sim1}
\end{table}

We compare the performance of a total of 17 different probabilistic forecasts (in the form of survival curves),  see Table~\ref{tab:sim1} where all methods are listed. In the first step of the combination procedure, the single-source forecasts may be estimated non-parametrically, using the Kaplan-Meier (KM) estimator, or using maximum likelihood (ML) with the true parametric model (the log-normal distribution, see Section~\ref{sec:simframe}). In this simulation set-up, combining the parametric curves will always be more advantageous (since the true model is known) and we have therefore not pursued the non-parametric option here, except for the HB method which is always non-parametric. 

In the second step of the combination procedure we can investigate two different ways to estimate the combination parameters: maximum likelihood (for LP, GP and BP), as in Section~\ref{sec:combML}, or minimising the IBS, as explained in Section~\ref{sec:minIBS}. 

We include four different versions of the GP method, and two versions of BP. The methods denoted by GP$_3$ and BP$_3$ in Table~\ref{tab:sim1} are the default versions of the methods as described in Section~\ref{sec:combometh}, each with three combination parameters to estimate. GP$_1$ and GP$_2$ are natural simplifications of GP$_3$: for GP$_1$ the only parameter left to estimate is $\omega$; the mean is fixed to 0 and the variance is fixed to 1. While for GP$_2$ only the mean parameter is fixed, leaving two combination parameters to estimate ($\omega$ and $\sigma$). BP$_2$ is a similar simplification of BP$_3$ where the beta parameters are constrained to be equal, which fixes the mean of the beta distribution to 1/2.
The method denoted by GP$_3$-t is the student-t version of the ordinary GP, as presented in Eq.~\eqref{eq:GPt}.

Finally, we also include two natural benchmark forecasts. The first one, LP$_0$, is simply the linear pool with the combination weight fixed at 0.5. So this methods requires no combination parameters to be estimated, and always gives equal weights to the two sources. The second one, `merge', consists of simply concatenating or merging the two datasets and then fitting a log-normal model to this single, larger dataset.    

The complete evaluation results are presented in the appendix, but the results for the methods in bold-face in Table~\ref{tab:sim1} will be given particular attention in Section~\ref{sec:simresults} below.

\subsection{Simulation model}
\label{sec:simframe}

We have considered the simulation framework carefully in order to have a suitable benchmark for comparing combination methods in a time-to-event setting. In particular, we have taken care to construct the framework so that the two sources contain different amounts of information, but that they both can produce calibrated forecasts.

The true realised time-to-hard-freeze in year $i$ is assumed to come from a log-normal distribution, with mean $\xi_{i}$ and variance $\tau_0^2$ on the log-transformed scale. Here we let  $\xi_i$ be a function of two variance components $x_{1,i}$ and $x_{2,i}$ which are realisations from normals with mean 0 and variances $\tau_1^2$ and $\tau_2^2$,
\begin{align*}
    &\text{Truth: } T_{i} | (x_{1,i},x_{2,i}) \sim \lnorm(\xi_i, \tau_0^2) \\
    &\text{with } \xi_i = \xi_0 + x_{1,i} + x_{2,i}, \\
    &\text{and } X_{1,i} \sim \N(0,\tau_1^2) \quad X_{2,i} \sim \N(0,\tau_2^2).
\end{align*}
The random variables $X_{1,i}$ and $X_{2,i}$ reflect some characteristic of each year: for example if both $x_{1,i}$ and $x_{2,i}$ are small there is a high chance of early frost in year $i$, while if both are large there is a higher chance of late frost.

For every year $i$, we generate two sets of forecasts which aim at predicting $T_i$: $n_1$ samples from source 1 and $n_2$ samples from source 2.
The data in source 1 are generated with knowledge of $x_{1,i}$, but not of $x_{2,i}$, while source 2 knows $x_{2,i}$, but not $x_{1,i}$. The log-normal has the property that the data from each source is still log-normally distributed,
\begin{align} \label{eq:sources}
    &\text{source 1: } Y_{i,j} | x_{1,i}  \sim \lnorm(\xi_0 + x_{1,i}, \tau_0^2 + \tau_2^2), \\
    &\text{source 2: } Z_{i,j} | x_{2,i}  \sim \lnorm(\xi_0 + x_{2,i}, \tau_0^2 + \tau_1^2) \nonumber
\end{align}
for $j=1,\dots,n_1$ for source 1 and $j=1,\dots,n_2$ for source 2.
The relative quality of the two sources depends on the sizes of $\tau_0$, $\tau_1$ and $\tau_2$. Intuitively, if $\tau_2 > \tau_1$ and large compared to $\tau_0$, then source 2 is much more informative than source 1.

Further, we let the two sources have different lead times, which translates into different censoring rates for the samples in the ensembles (see Section~\ref{sec:motivation}). In line with our application in Section~\ref{sec:application} we let source~(1) have a long time, equal 120 days into the future, leading to a low censoring rate. While source~(2) has a short lead time of 60 days, and therefore a higher censoring rate.

\subsection{Scenarios}
\label{sec:simscen}

We consider a total of 16 simulation scenarios, see Table~\ref{tab:sim2}. 
The number of training observations, i.e. the number of past forecast-observation pairs $n$, can be expected to have a large influence on the estimation of the combination parameters.
Scenarios 1 to 8 are idealised scenarios where the combination parameters are estimated based on a large number of training years, $n=1000$.  
 The training data consists of realised times-to-frost as well as forecasts from the two sources belonging to each realisations. From this training sample we estimate the combination parameters. Then, we generate a large ($10^4$) test sample on which the single-source and combined forecasts are evaluated. 
Scenarios 9 to 16 are more realistic in the sense that we have a small number of training years, $n=20$, from which the combination parameters are estimated. Each set of combination parameters is then evaluated on a single test dataset, and this whole procedure is repeated a large ($10^4$) number of times. 

\begin{table}[ht]
\caption{\small Simulation scenarios. }
\centering \small
\begin{tabular}{c|cccc} \toprule
    Scenario & No. of training years & Balanced sources? & Bias? & Ensemble sizes \\ \midrule 
  1 & $n=1000$ &  yes & no & $n_1=100$, $n_2=20$
  \\ 
  2 & $n=1000$  &  yes & no & $n_1=n_2=20$  \\ 
  3 & $n=1000$ &  yes & yes & $n_1=100$, $n_2=20$  \\ 
  4 & $n=1000$  &  yes & yes & $n_1=n_2=20$  \\ 
  5 & $n=1000$ & no & no & $n_1=100$, $n_2=20$ \\
  6 & $n=1000$ & no & no & $n_1=n_2=20$ \\
  7 & $n=1000$ & no & yes & $n_1=100$, $n_2=20$ \\
  8 & $n=1000$ & no & yes & $n_1=n_2=20$ \\
  9 & $n=20$ &  yes & no & $n_1=100$, $n_2=20$ \\ 
  10 & $n=20$ &  yes & no & $n_1=n_2=20$  \\ 
  11 & $n=20$ &  yes & yes & $n_1=100$, $n_2=20$  \\ 
  12 & $n=20$ &  yes & yes & $n_1=n_2=20$ \\ 
  13 & $n=20$ & no & no & $n_1=100$, $n_2=20$ \\
  14 & $n=20$ & no & no & $n_1=n_2=20$ \\
  15 & $n=20$ & no & yes & $n_1=100$, $n_2=20$ \\
  16 & $n=20$ & no & yes & $n_1=n_2=20$ \\\bottomrule
\end{tabular}
\label{tab:sim2}
\end{table}

The scenarios also differ with respect to characteristics of the two sources. For some scenarios we let both sources have small samples, $n_1=n_2=20$, while in the others we strive to mimic the application in Section~\ref{sec:application} and let source 1 have a large ensemble ($n_1=100$), while source 2 has a small one ($n_2=20$). Note here that source 2 also has a higher censoring rate in all scenarios, so the effective sample size is even smaller. Further, we sometimes let the two sources be 
balanced in terms of information content, while in other scenarios we let the sources be 
very unbalanced. In the first case we have $\tau_0=\tau_1=\tau_2=0.4$,  and in the second case we let source 2 have much less information than source 1, with   $\tau_0= 0.53$, $\tau_1=0.4$, and $\tau_2=0.2$. 
Finally, we include the possibility of sources (2) producing biased forecasts. In these scenarios we have
\begin{align*} \label{eq:sources}
    &\text{source 2: } Z_{i,j} | x_{2,i}  \sim \lnorm(\xi_0 + x_{2,i} - 0.5, \tau_0^2 + \tau_1^2),
\end{align*}
and source 1 as before.

 \subsection{Results}
\label{sec:simresults}
As stated in Section~\ref{sec:eval} we will evaluate forecasts with respect to sharpness under calibration. Calibration is assessed by PIT histograms
and by computing the mean and standard deviation of the PIT values.
The predictive performance of the forecasts is assessed by the average integrated Brier score (IBS) across all test observations. With the IBS we use $t_{\max}=120$ here, see Eq.~\eqref{eq:ibs}.

We start by presenting the results from the idealised scenarios with a very large training dataset ($n=1000$), before we present the results for the more realistic scenarios where the training dataset is small ($n=20$). In the figures below, we present the results for a selection of the methods listed in Table~\ref{tab:sim1}. The complete evaluation results are provided in the appendix. 

\subsubsection{Large training dataset}

 \begin{figure}[h] 
\begin{center}
\includegraphics[scale=0.45]{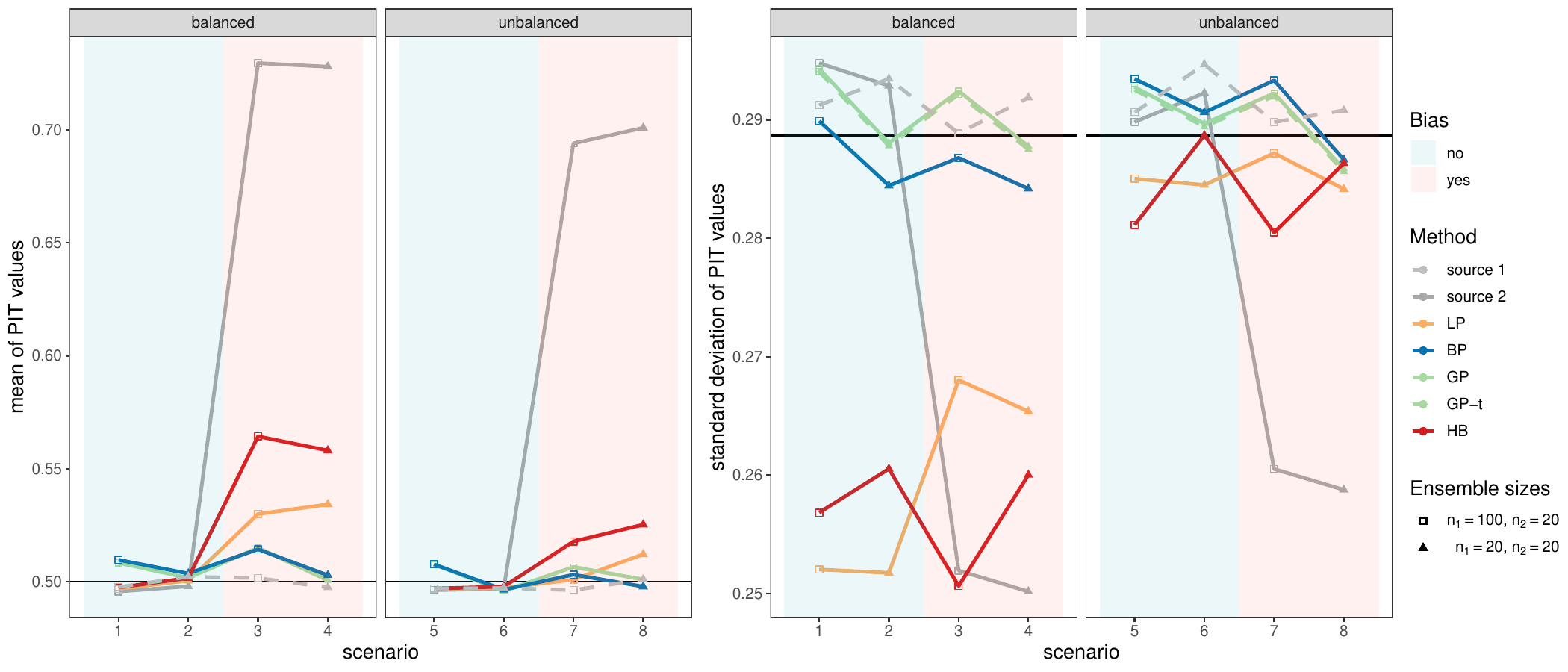}
\end{center}
\caption{\small Mean (left) and standard deviation (right) of PIT values for a selection of methods in the 8 scenarios with large training dataset ($n=1000$).
Perfect calibration is indicated by the straight black lines, at 0.50 for the mean and at $\sqrt{1/12}$ for the standard deviation. 
Scenarios with bias have a red background, while scenarios without bias have a blue background.  Scenarios with larger ensembles are shown by rectangles, and scenarios where both ensembles are small are shown by triangles.
}
\label{fig:sim1}
\end{figure}

 \begin{figure}[h] 
\begin{center}
\includegraphics[scale=0.45]{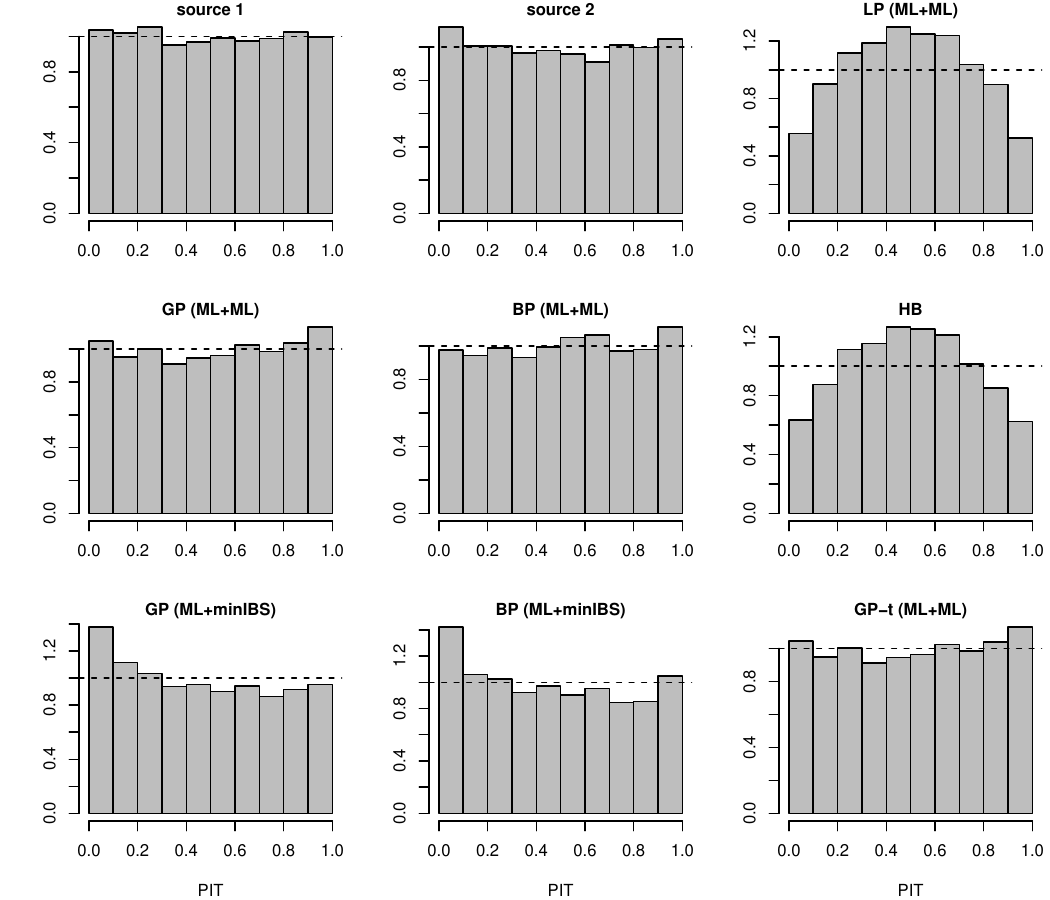}
\includegraphics[scale=0.45]{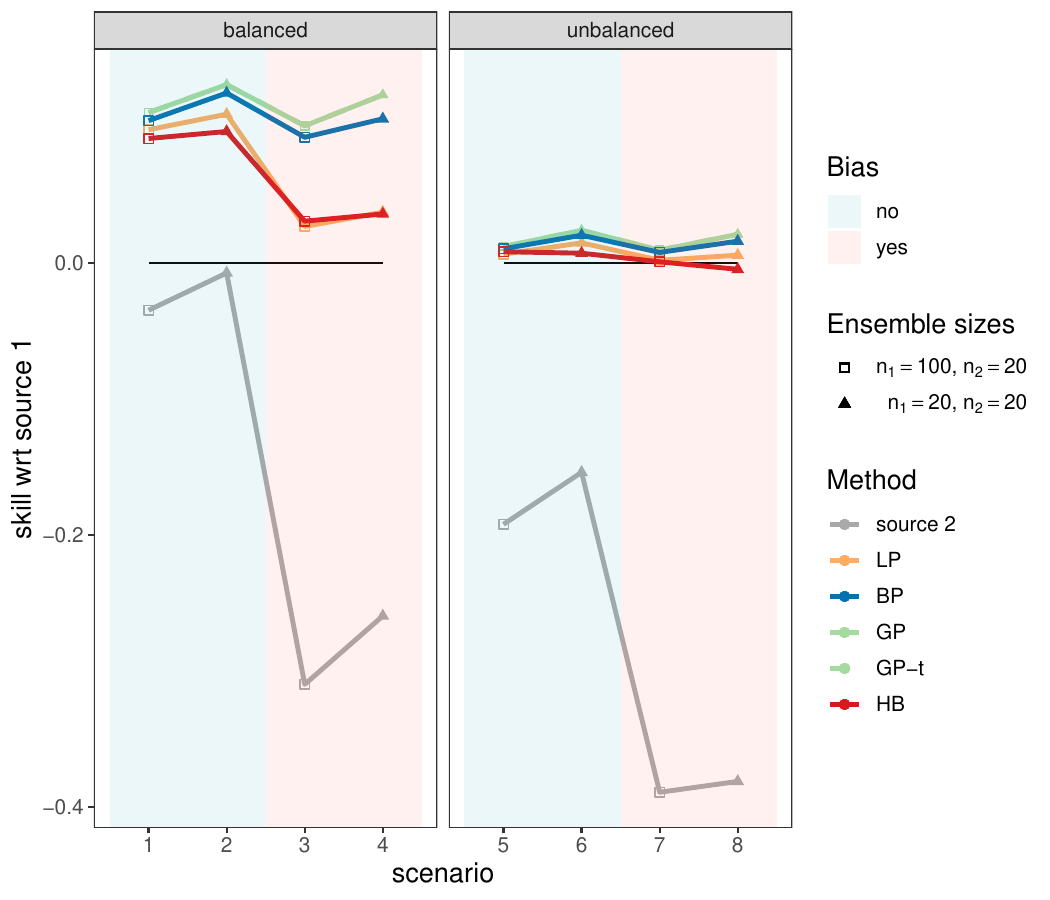}
\end{center}
\caption{\small Left: PIT histograms for scenario 1 (large training dataset, balanced sources and no bias) for the single-source forecasts and the combined forecasts from methods LP, BP$_3$, GP$_3$, HB, BP$_{\ibs}$, GP$_{\ibs}$ and GP$_3$-t. \\
Right:  IBS skill scores for a selection of methods in the 8 scenarios with large training dataset ($n=1000$). The straight black line at 0 gives the skill of the forecast from source 1: forecasts that lie above the line have better predictive performance than source 1. Scenarios with bias have a red background, while scenarios without bias have a blue background.  Scenarios with larger ensembles are shown by rectangles, and scenarios where both ensembles are small are shown by triangles. }
\label{fig:sim2}
\end{figure}

The left panel of Figure~\ref{fig:sim1} displays the mean of the PIT values for several methods across all scenarios with large training dataset. In scenarios where both sources are unbiased, all combination forecasts are unbiased too, i.e. with a mean PIT value close to 0.50. In scenarios where source 2 is biased we see that the simple combination methods LP and HB tend to produce biased forecasts too, though with less bias than the forecast from source 2.

The right panel of Figure~\ref{fig:sim1} displays the standard deviation of the PIT values, and demonstrates that the simple combination methods LP and HB  fail to produce calibrated forecasts, as expected from the existing literature on forecast combination (see Section~\ref{sec:combometh}). Even in scenarios without bias, these methods produce overdispersive forecasts, which manifest themselves in PIT standard deviations well below 0.288. The overdispersive forecasts of LP and HB are also clearly visible in the PIT histograms shown in the left panel of Figure~\ref{fig:sim2}

Skill scores (derived from the IBS scores) with respect to the forecasts from source 1 are displayed in the right panel of Figure~\ref{fig:sim1} for a subset of methods. A positive skill score indicates that the forecast has a better predictive performance than source 1. 
The forecast from source 2 is generally inferior to the forecast from source 1, by design. source 2 obtains a performance on par with source 1 only in scenarios with balanced sources and no bias. 

The predictive performance of the combination forecasts is strongly influenced by the degree of imbalance between the sources. When the sources are balanced, all combination methods greatly outperform the single-source forecasts. In these scenarios the more complex combination methods, BP$_3$, GP$_3$ and GP$_3$-t, also outperform the simpler combination methods (LP and HB), particularly when there is bias in one source. 
When the sources are unbalanced, the combination forecasts still  outperform the single-source forecasts in most scenarios, but not considerably. Also, the more complex combination methods only have a slightly higher skill scores than the simpler ones.

\subsubsection{Small training dataset}

 \begin{figure}[h] 
\begin{center}
\includegraphics[scale=0.45]{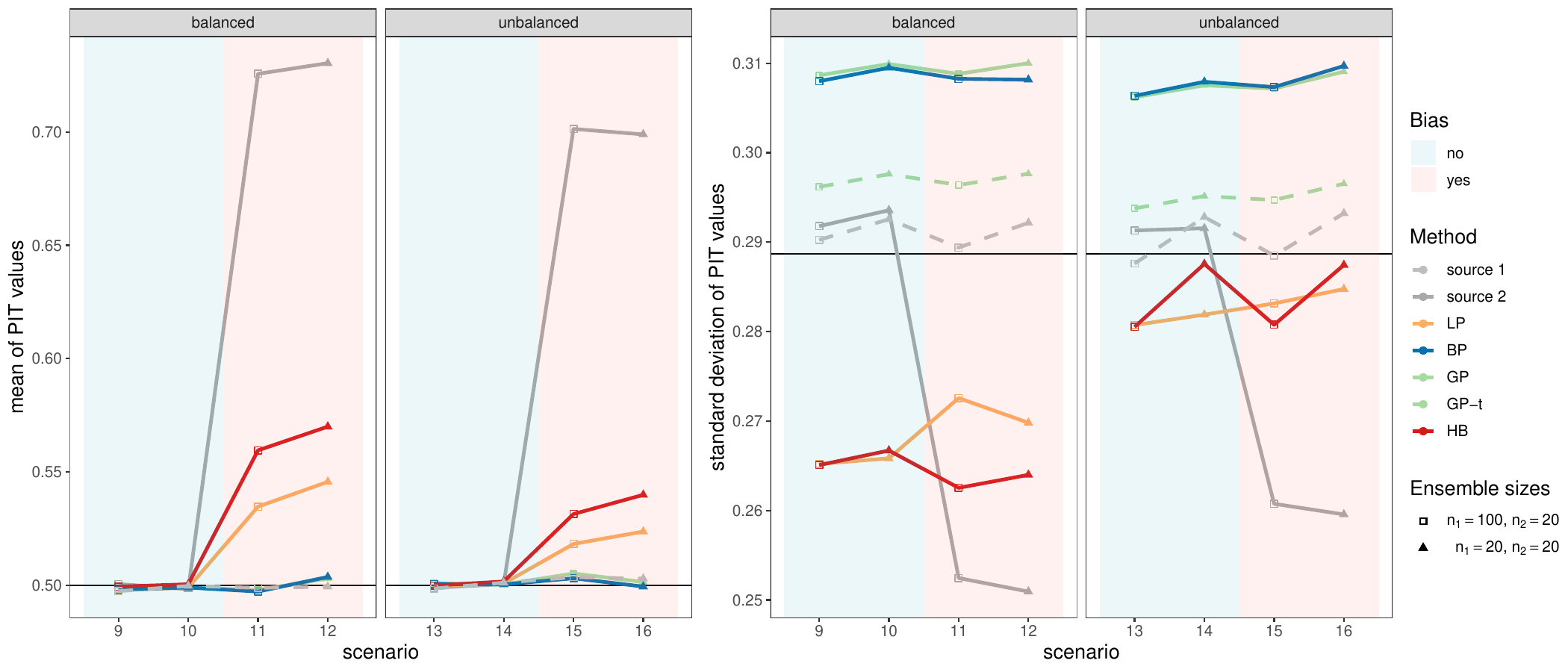}
\end{center}
\caption{\small Mean (left) and standard deviation (right) of PIT values for a selection of methods in the 8 scenarios with small training dataset ($n=20$).
Perfect calibration is indicated by the straight black lines, at 0.50 for the mean and at $\sqrt{1/12}$ for the standard deviation. 
Scenarios with bias have a red background, while scenarios without bias have a blue background.  Scenarios with larger ensembles are shown by rectangles, and scenarios where both ensembles are small are shown by triangles.
}
\label{fig:sim3}
\end{figure}

 \begin{figure}[h] 
\begin{center}
\includegraphics[scale=0.45]{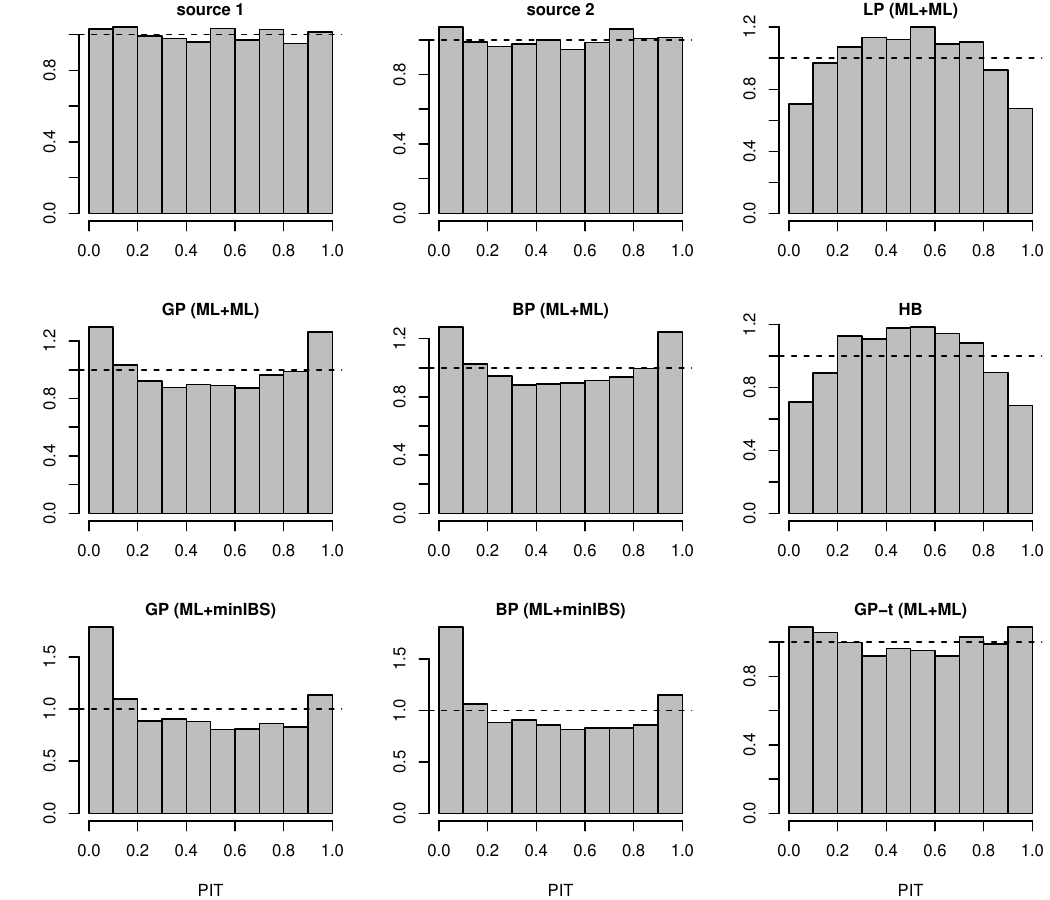}
\includegraphics[scale=0.45]{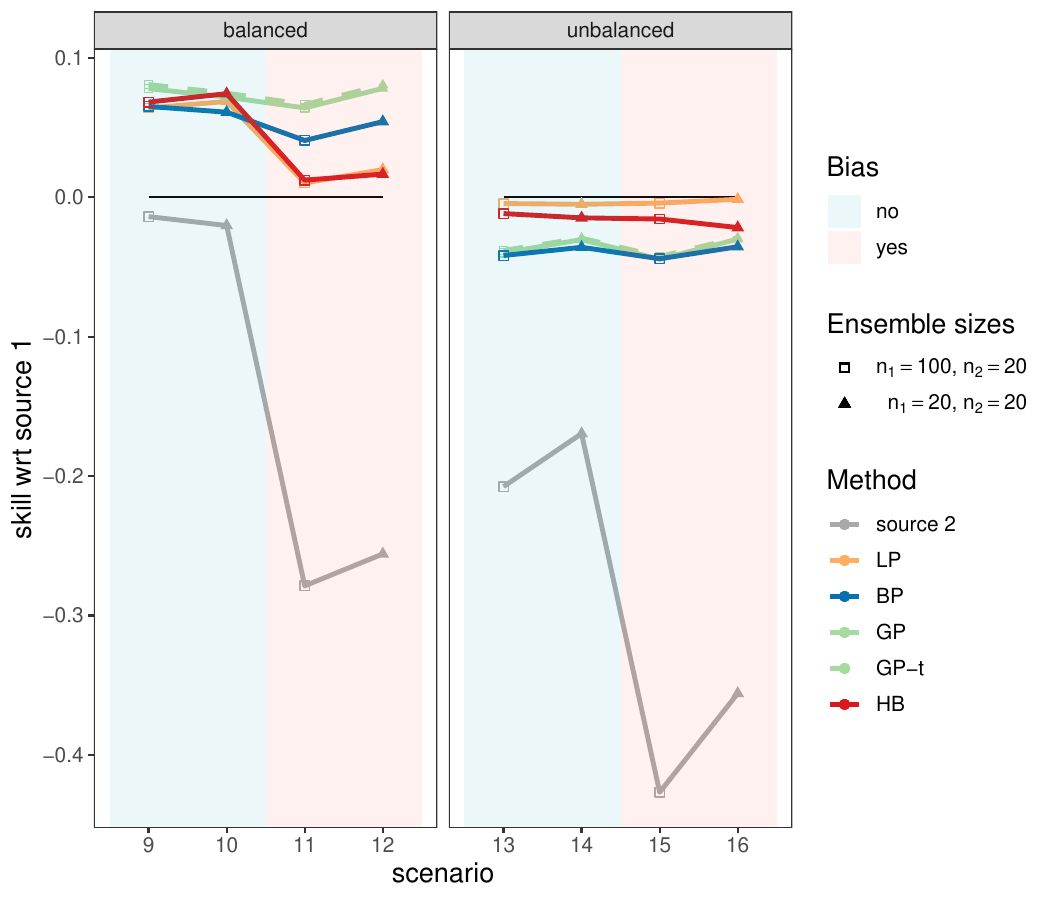}
\end{center}
\caption{\small Left: PIT histograms for scenario 9 (small training dataset, balanced sources and no bias) for the single-source forecasts and the combined forecasts from methods LP, BP$_3$, GP$_3$, HB, BP$_{\ibs}$, GP$_{\ibs}$ and GP$_3$-t. \\
Right:  IBS skill scores for a selection of methods in the 8 scenarios with small training dataset ($n=20$). The straight black line at 0 gives the skill of the forecast from source 1: forecasts that lie above the line have better predictive performance than source 1. Scenarios with bias have a red background, while scenarios without bias have a blue background.  Scenarios with larger ensembles are shown by rectangles, and scenarios where both ensembles are small are shown by triangles. }
\label{fig:sim4}
\end{figure}

Figure~\ref{fig:sim3} displays the calibration results for several methods in the scenarios with small training dataset ($n=20$). The figure with the mean PIT values is roughly similar to the corresponding figure for the large training dataset scenarios. In the figure with the standard deviation of the PIT values (right panel) we see that the forecasts from the simple combination methods LP and HB are still overdispersive, but somewhat less so than in the large $n$ scenarios. Further, the more complex combination methods, BP$_3$ and GP$_3$, now produce underdispersive forecasts, with standard deviations clearly above 0.288. In these scenarios, the calibration of the GP forecast is greatly improved by using the student-t version GP$_3$-t. This effect is also visible by comparing the PIT histograms for GP$_3$ and GP$_3$-t in the left panel of Figure~\ref{fig:sim4}.

The IBS skill scores in the right panel of Figure~\ref{fig:sim4}  demonstrate that when the sources are balanced, all combination forecasts still outperform the single-source forecasts. When the sources are unbalanced however,  all combination methods have worse predictive performance than the forecast from source 1. In those scenarios, the simpler method are better than the more complex one, but do not generally beat source 1 nevertheless.

\subsubsection{Lessons from simulations}

The main finding in these simulations is that combination methods of the type we study here can be beneficial, in the sense that they produce forecasts which are calibrated and have higher predictive skill than the single-source forecasts. However, the combination methods require either sufficiently balanced sources or sufficient amounts of training data. In such situations, more complex combination methods will be particularly beneficial, both in terms of skill and calibration. If the sources are unbalanced \textit{and} training data is scarce all combination methods struggle, particularly the more complex ones.


The simulations also illustrate various other lessons. First, the parametric single-source forecasts outperform the non-parametric ones in terms of predictive performance (see Table~\ref{tab:appIBS}). This is not surprising as the parametric single-source forecasts make use of the true model from which the data was generated. 

Second, we see that most combination methods generally beat the two benchmark methods, in terms of predictive performance and calibration, see tables in Appendix~\ref{app:SimResults}.  The equal-weight method LP$_0$ has particularly bad predictive performance in scenarios with bias or unbalanced sources. The naive method `merge' fares a bit better, particularly in scenarios with little training data, but is still worse than the proper combination methods in terms of overall predictive performance, see Table~\ref{tab:appIBS}. Both benchmark methods produce biased forecasts in the scenarios with bias, see Table~\ref{tab:appExp}, and overdispersive forecasts in almost all scenarios, see Table~\ref{tab:appSD}.

Third, the simulations allow us to compare methods where the combination parameters are estimated by maximum likelihood (ML) and the same methods with combination parameters estimated by minimising the IBS (minIBS). In terms of predictive performance, estimation by minIBS is roughly similar to estimation by ML. For some methods, minIBS is often the best, while for other methods ML is often a bit better, see Table~\ref{tab:appIBS}. The estimation by minIBS can produce less calibrated forecasts however, see Figures~\ref{fig:sim2} and \ref{fig:sim4} where the PIT histograms for BP$_{\ibs}$ and GP$_{\ibs}$ reveal some degree of bias. Also, estimation by ML is much faster, particularly for the more complex combination methods.

In addition, we can compare different versions of the more complex combination methods, GP and BP, see Table~\ref{tab:appIBS}. The one-parameter method GP$_1$ achieves predictive performance on par with the other one-parameter methods LP and HB (even beating both of these in some scenarios). However it suffers from the same issues as these methods with respect to calibration: overdispersive forecasts in the unbiased scenarios and biased forecasts in the scenarios with bias in source 2. The two-parameter methods GP$_2$ and BP$_2$ do not produce overdispersive forecasts (since they can correct the variance), but get biased forecasts in the scenarios with bias. These two methods could then be good choices in applications where the sources are known, or assumed, to be unbiased. These results illustrate the effect of the extra combination parameter in the more complex combination methods. For GP for example, the mean parameter fixes potential biases and the variance parameter ensures calibration.

Finally, we learn that ensemble size does not tend to have a large effect. At least not for the set of values we have studied here.

\section{Case-study: time-to-hard-freeze}
 \label{sec:application}

 To test the combination methods in practice, we use a dataset consisting of seasonal and subseasonal temperature forecasts that give information about the first occurrence of hard freeze after October 1 for a set of locations in Norway and Fennoscandia. Here, hard freeze refers to a mean daily near-surface air temperature below 0$^\circ$C, as described in Section~\ref{sec:motivation}. The goal is to test whether a combination of subseasonal and seasonal forecasts gives more accurate time-to-hard-freeze forecasts than any of the two forecasting products alone, and investigating differences between the combination methods proposed in Section \ref{sec:combometh}.

\subsection{Data}
We use daily temperature data from three different sources: (1) seasonal forecasts initialized on September 1; (2)~subseasonal forecasts initialized on September 30; and (3) the actual time-to-hard-freeze observations which we obtain from the gridded, observation-based data product SeNorge2018 (version 18.12). The forecasts and observations cover the period from 1999 to 2018, which will be our study period.

Seasonal forecasts of 2-meter temperatures, initialized on September 1, were downloaded from the Copernicus Climate Data Store (CDS) \footnote{\url{https://confluence.ecmwf.int/display/CKB/C3S+Seasonal+Forecasts\%3A+dataset+documentation}} for 1999-2018 and come from the NWP models of 5 different weather centres (ECMWF, CMCC, DWD, Méteo France and UK Met Office). The number of ensemble members varies between years, but for each year there are approximately $150$ ensemble members in total, with a maximum lead time of 5-6 months. See Appendix~\ref{app:DataSources} for more details about the forecasting systems involved.

The subseasonal forecasts are 2-meter temperature NWP forecasts delivered by ECMWF, initialized on September 30 \footnote{\url{ https://confluence.ecmwf.int/display/S2S/ECMWF+model+description}}. We use a forecast produced in 2019, with forecast cycle number CY46R1. Forecasts with issue dates in 1999-2018 were available with 11 ensemble members and a maximum lead time of 46 days (November 15).

The SeNorge product is a dataset developed by the Norwegian Meteorological Institute and contains mean daily temperature data for Norway and parts of Fennoscandia \citep{lussana2018}. The product is based on 2-meter temperature observations from meteorological stations that are interpolated to a 1 km $\times$ 1  km grid by using the optimal interpolation approach \citep{Gandin,Kalnay}. See \citet{lussana2018} for details. Here, we use the SeNorge data as a substitute for local, observed surface temperatures, as in \cite{roksvaag2023}. 

Among the large number of grid nodes available in the SeNorge dataset, study locations are selected so that seasonal and subseasonal forecasts are available in reasonable proximity. We only select locations with mean elevation between 0.5 and 800 m a.s.l. The elevation restriction is added since the study is motivated by applications in agriculture and we therefore want to avoid grid nodes that mainly cover sea, lakes or mountains. Also, locations that experience the first hard freeze before October 1 for more than 5 out of our 20 study years are omitted from the analysis. See more details about the selection of study locations in  Appendix~\ref{app:DataSources}.
The resulting 103  study locations are shown in Figure~\ref{fig:mean_obs_ttf}, together with histograms of the observed or forecasted time-to-hard-freeze for the three datasets.

\begin{figure}[h] 
\begin{center}
\includegraphics[scale=0.5]{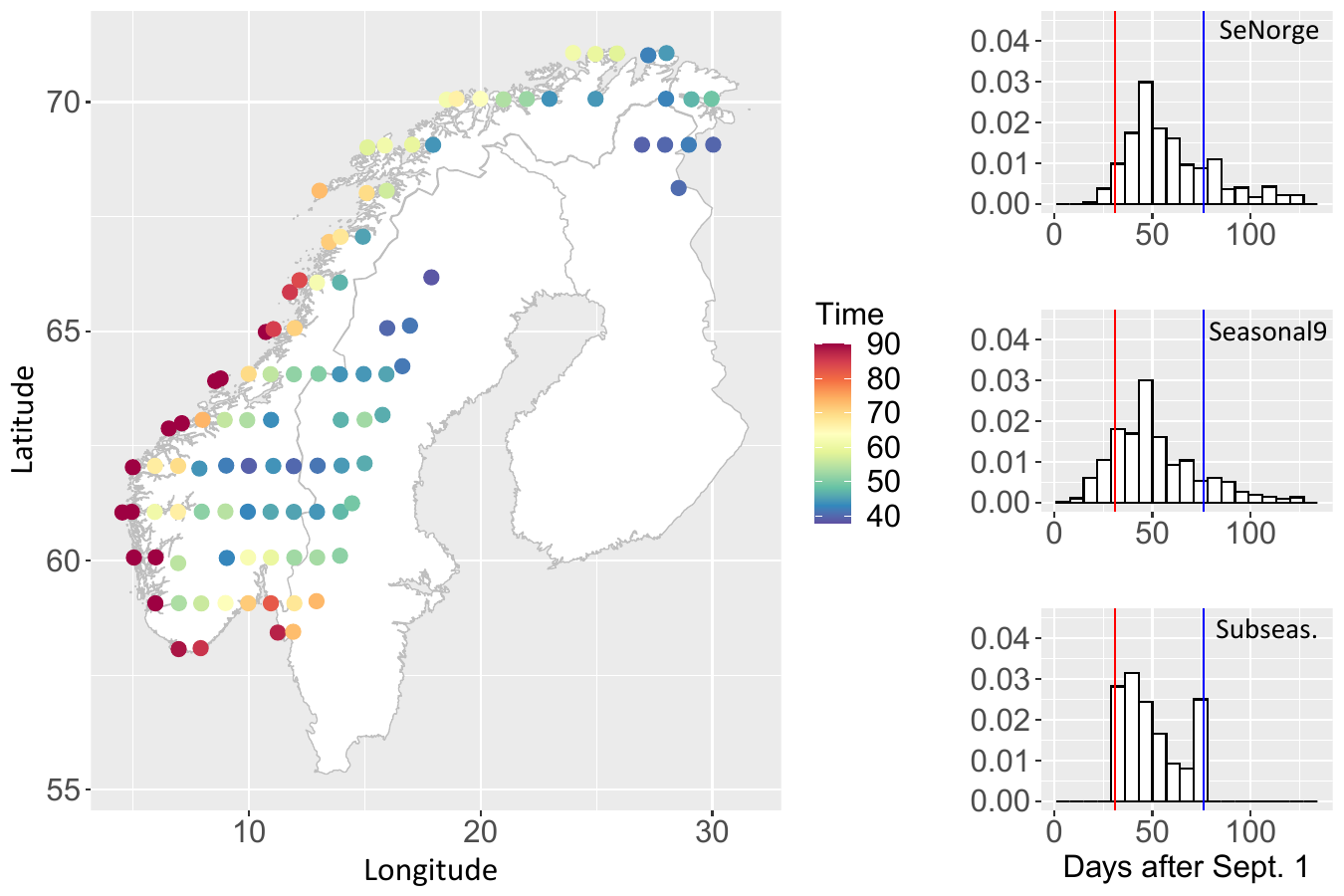}\\
\end{center}
\caption{\small Left: The 103 study locations in Fennoscandia and their mean observed time-to-hard-freeze from September 1 (1999-2018) according to SeNorge, where the scale is truncated at 122 days (December 31). Right: Distribution of the observed or forecasted time-to-hard-freeze for all locations for 1999-2018 for SeNorge, the post-processed seasonal forecast ensemble and the post-processed subseasonal forecast ensemble. The sample sizes are 2060, 355212 and 30360 for the three data sources. The red line marks October 1 (day 31), while the blue line shows the maximum lead time date of the subseasonal forecast (day 30+46, November 15).}
\label{fig:mean_obs_ttf}
\end{figure}

Seasonal and subseasonal forecasts are known to be subject to systematic biases \citep{ecmwfeval,postprocessingogsubseas,bias,drift}, and also come at a considerably coarser resolution compared to SeNorge. To make sure that the forecasted temperatures better match the temperatures observed at the target locations, the raw seasonal and subseasonal forecasts are post-processed in the same manner as in \cite{roksvaag2023}, by adjusting the forecast means and standard deviations relative to historical means and standard deviations from SeNorge, see more details in Appendix~\ref{app:PostProc}. From the post-processed seasonal and subseasonal temperature forecasts, we extract the day of the first occurrence of hard freeze after September 1 for each combination of year, location, ensemble member and forecast system. Likewise, we take out the first day of observed hard freeze for SeNorge.

\subsection{Experimental set-up}
We use a subset of the combination methods from Table \ref{tab:sim1} to predict the time to the first autumn frost from October 1 by combining forecasts produced from the two data sources:

\begin{itemize}
\item Source 1: the seasonal forecast, i.e. an ensemble of around $150$ members with issue date September 1 and censoring date December 31.
\item Source 2: the subseasonal forecasts, i.e. an ensemble of 11 members with issue date September 30 and censoring date November 15.
\end{itemize}

From the above data sources, we make single-source forecasts by using the non-parametric Kaplan-Meier estimator (KM) or a parametric log-normal distribution for each combination of year and location, independently of other years and locations. We consider October 1 as the issue date of the combined forecast.  To account for what has happened during the month of September, we remove ensemble members in the seasonal forecast that predict hard freeze before October 1 if there are any. 

After estimating the single-source forecast for a location, we go through each target year (1999-2018) and estimate the parameters of the combination methods based on forecasts from the remaining 19 hindcast years. There are then 19 forecast-observation pairs available for the estimation of the combination parameters. 
For this estimation step we consider both non-parametric and parametric methods, see Section~\ref{sec:estimation}.
In the non-parametric case, with single-source KM forecasts, we test the Hazard Blending method (HB) and Linear Pooling (LP). Combination parameters are estimated by minimizing the IBS. In the parametric case, with single-source log-normal forecasts, we fit the Linear Pooling method (LP), the Gaussian pooling (GP) with 1 to 3 parameters (GP$_1$,  GP$_2$, GP$_3$) and the Beta Pool (BP) method with 2 and 3 parameters (BP$_2$, BP$_3$). Combination parameters are then estimated by maximum likelihood (ML). In addition, we test two benchmark methods, the LP$_0$ method, i.e. linear pooling with fixed, equal weights on both forecast sources ($\omega=0.5$), and the `merge' method where a log-normal model is fitted to a combined ensemble with members from both  seasonal and subseasonal forecasts.

The combination forecasts are evaluated by using the IBS from Eq.~\eqref{eq:ibs} computed from October 1 (day 31) to December 31 (day 122), with the time-to-hard-freeze from SeNorge as the ground truth. We calculate location-specific IBS', based on the 20 (annual) predictions performed for each location. In the case that the first hard freeze has already occurred at the target location and target year before October 1, or on October 1, we set the IBS to 0 for all methods. This happens for 82 of the 2060 locations-year combinations ($103 \times 20$) that we are considering.

The single-source subseasonal and seasonal forecasts are used as reference forecasts, both non-parametric (KM)  and parametric (log-normal, ML) versions. In addition, we make a climatology reference forecast by fitting the KM estimator to the historical observed time-to-hard-freeze observations from 1999-2018, with the target year being left out together with years that give hard freeze before October 1 (if any). Results based on both raw and post-processed ensembles are presented.

\subsection{Results}

The results from the case-study are summarized in Table \ref{tab:IBSrealdata} in terms of the location-specific IBS scores averaged over all study locations. We see that most combination methods beat the single-source forecasts, and also the climatology, when post-processed ensembles are used as input. The simpler combination methods generally perform better than the more complex combination methods. The LP (ML) method has the lowest IBS of all methods, while the methods GP$_1$, GP$_2$ and BP$_2$  have lower mean IBS' than their more complex counterparts GP$_3$ and BP$_3$. This is likely due to the low number of hindcast years (19) which may result in overfitting and unstable parameter estimates, particularly for the methods with several combination parameters. 

The best combination methods have somewhat better predictive power than the naive benchmark methods LP$_0$ and `merge'. Note also that the log-normal model seems to fit the forecast ensembles well: we see that the predictive performance of the single-source forecasts is actually higher when using the log-normal model (ML) than with the non-parametric estimator (KM). This is the case for both sources, but the effect is particularly large for source (2) which has few ensemble members.

\begin{table}[ht]\caption{\small  Mean IBS for the combination methods (LP, GP, BP, HB, merge) and the reference single-source forecasts (Clim, Seas9 and Subseas) over all target locations. Models marked with KM are based on the non-parametric Kaplan-Meier survival function and estimated by minimizing the IBS, while models with ML are based on the parametric log-normal survival function and fitted using the maximum likelihood estimator (ML). Results for both post-prosessed and raw forecasts are shown. Combination methods with a lower mean IBS than all of the three single-source forecasts (Clim, Seas9 and Subseas) are highlighted in bold in each row.}
\centering \footnotesize 
\begin{tabular}{c|ccc|cc|cccc} \toprule
 &   Clim (KM)  & Seas9 (KM) & Subseas (KM)  & Seas9 (ML) & Subseas (ML) & LP (ML) & GP$_3$ (ML) & BP$_3$ (ML) & HB (KM)  \\ \midrule 
Post-pr.& 0.0753 & 0.0758 & 0.1034 & 0.0748 & 0.0854 & \textbf{0.0723} &  0.0762 & 0.0760 & \textbf{0.0735}  \\ 
  Raw & 0.0753  & 0.0946 & 0.1396 & 0.0941 & 0.1239 & 0.0837 &  0.0772 & \textbf{0.0745} & 0.0829   \\ \bottomrule
\end{tabular}\vspace{3mm}
\begin{tabular}{c|cccc|ccc} \toprule
  &  LP (KM) & GP$_1$ (ML)& GP$_2$ (ML)& BP$_2$ (ML) & LP$_0$ (ML) & Merge (ML) \\  \midrule 
Post-pr. & \textbf{0.0733} & \textbf{0.0726} & \textbf{0.0735} & \textbf{0.0733} & \textbf{0.0747} & \textbf{0.0739}  \\
  Raw  &   0.0847 & 0.0838 & 0.0829 & 0.0808 & 0.0918 & 0.0912 \\ \bottomrule
\end{tabular} \label{tab:IBSrealdata}
\end{table}

Considering the raw forecasts, the combination methods with three parameters (GP$_3$ and BP$_3$) perform better than the simpler methods (GP$_1$,GP$_2$, BP$_2$, HB, LP). According to Table \ref{tab:IBSrealdata}, all the combination methods outperform the single-source forecasts they are based on, but BP$_3$ also beats the climatology forecast as the only method. The results suggest that the 1-2 extra parameters in the  BP$_3$ (or GP$_3$) model are more important when we have not already accounted for biases and dispersion through post-processing.

Figure \ref{fig:pit_realdata} displays PIT histograms for the predictions of the time-to-hard-freeze for the combination forecasts and for the single-source forecasts when using post-processed input data. Despite the post-processing, the subseasonal forecasts are clearly biased, with many PIT values around 1 suggesting underestimation of the time-to-hard-freeze. The seasonal forecast tends  towards overestimation in the ML case.  The combination methods based on parametric single-source forecasts (ML) appear to be better calibrated than the methods using non-parametric forecasts (KM). The non-parametric HB (KM) and LP (KM) both show signs of overestimation, while the method that performed best in terms of the IBS, LP (ML), looks well-calibrated. See Table \ref{tab:PITreal} for mean and standard deviation of the PIT values for all methods and Figure \ref{fig:PIThist_raw} for PIT histograms for the raw forecasts.

\begin{figure}[h] 
\begin{center}
\includegraphics[scale=0.73]{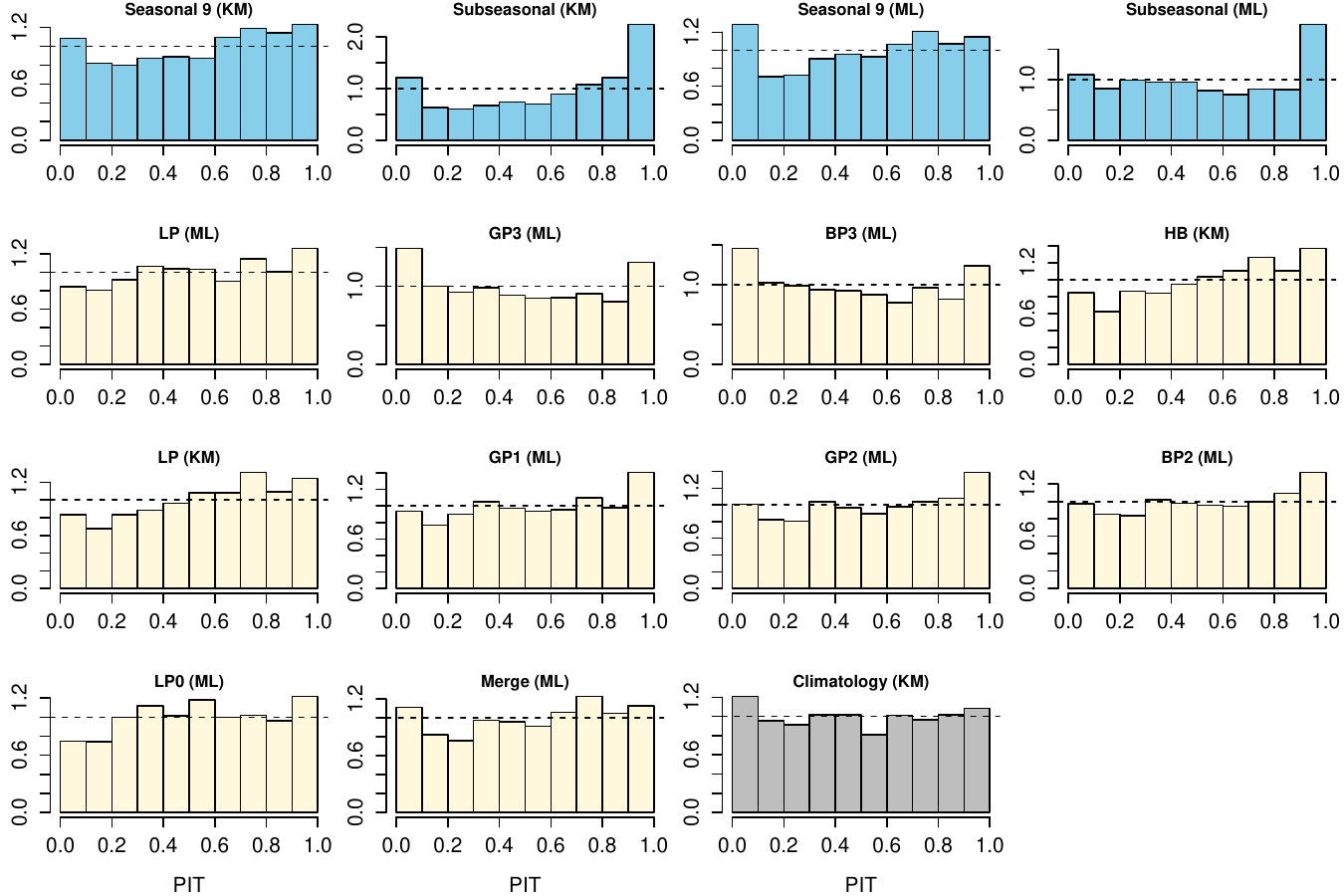}\\
\end{center} 
\caption{\small PIT histograms for the single-source forecasts (blue), 
the combination methods (yellow) and the climatology forecast (gray) for predictions of time-to-hard-freeze when post-processed forecasts are used.}
\label{fig:pit_realdata}
\end{figure}

In the remainder of this section, we focus on the best-performing combination method, LP (ML), for the post-processed forecasts. In Figure \ref{fig:IBS_LP_realdata} the IBS scores of LP (ML) are compared to the IBS' of the single-source forecasts. The right plot in the upper panel shows that the IBS of the subseasonal forecast is rather high compared to the climatology for many locations, but there are also locations where the subseasonal forecast gives better predictions than the climatology. Next, considering the right plot in the lower panel, we see how the combination method is able to improve the results at locations where the subseasonal forecast performed poorly, while at the same time keeping the good performance for locations where the subseasonal forecast was performing well. 

\begin{figure}[h] 
\begin{center}
\includegraphics[scale=0.5]{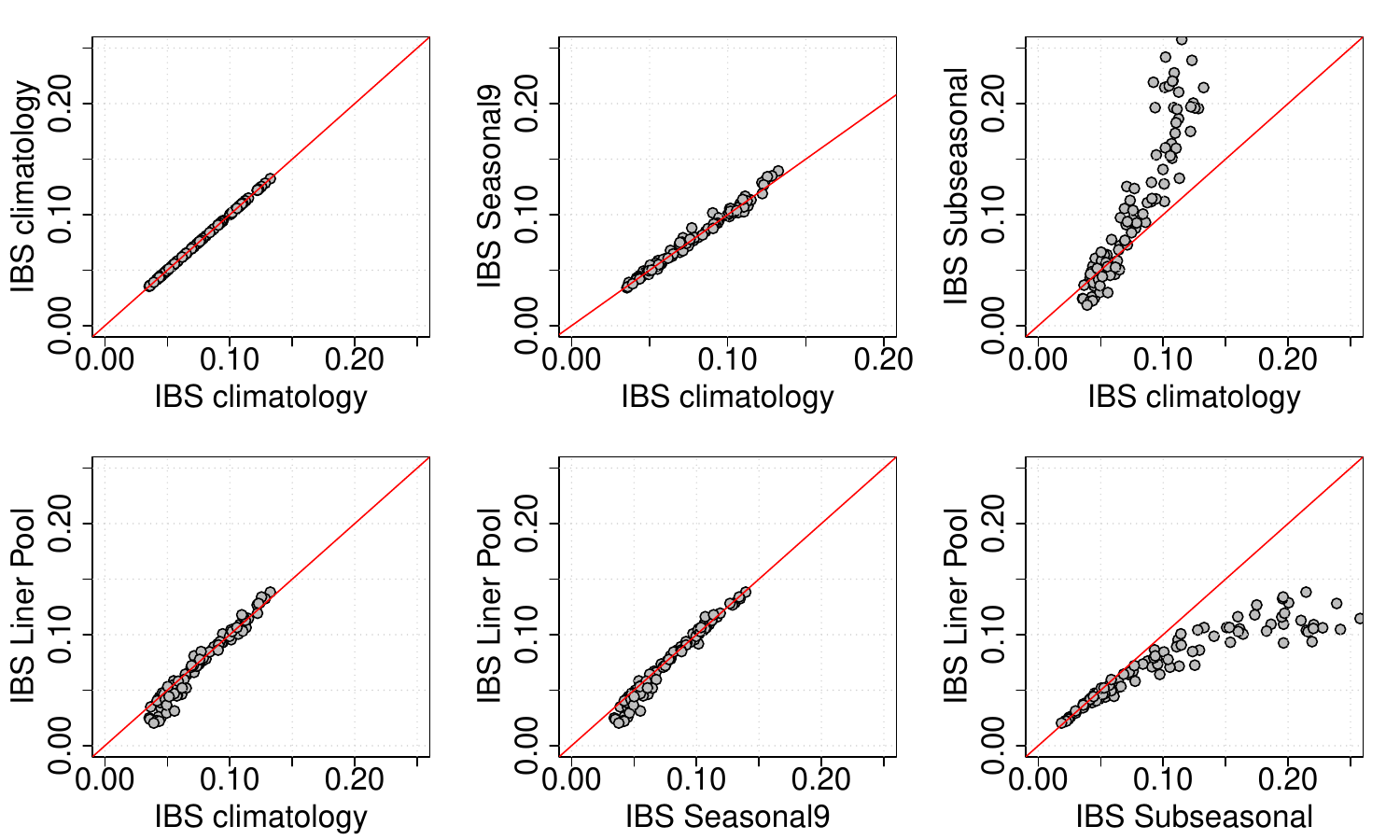}\\
\end{center}
\caption{\small The top row shows location-specific IBS' for the single source forecasts, Subseasonal (ML) and Seasonal (ML), compared to climatology forecast, Clim (KM). The bottom row shows the IBS' for the Linear Pool combination (ML) method that is  compared to the climatology forecast and the single source forecasts (bottom row). All ensemble forecasts involved are post-processed. }
\label{fig:IBS_LP_realdata}
\end{figure}

Figure \ref{fig:IBS_LP_realdata} further show that compared to the seasonal forecast, the combination method is better or approximately as good for most locations (lower panel, centre plot). The same applies for the combination method compared to climatology (lower panel, left plot). We also notice how the climatology forecast and the seasonal forecasts are almost equivalent in performance (upper panel, centre plot), which is in line with the results in \citep{roksvaag2023}: the skill of the seasonal time-to-hard-freeze forecast is typically high compared to climatology for regions where hard freeze occurs the first three weeks after the forecast issue date, and in this study, most locations experience hard freeze on day 31 or later compared to the issue date of the seasonal forecast (see Figure \ref{fig:mean_obs_ttf}). Currently, the seasonal forecasts we use in this study are only initialized once a month (September 1, October 1,...) and published 12 days after their initialization date (September 12, October 12,...). Ideally, a seasonal forecast from mid-September, published just before October 1, would be available in real time for combination.

\begin{figure}
     \centering
     \begin{subfigure}[b]{0.30\textwidth}
         \centering
         \includegraphics[width=\textwidth]{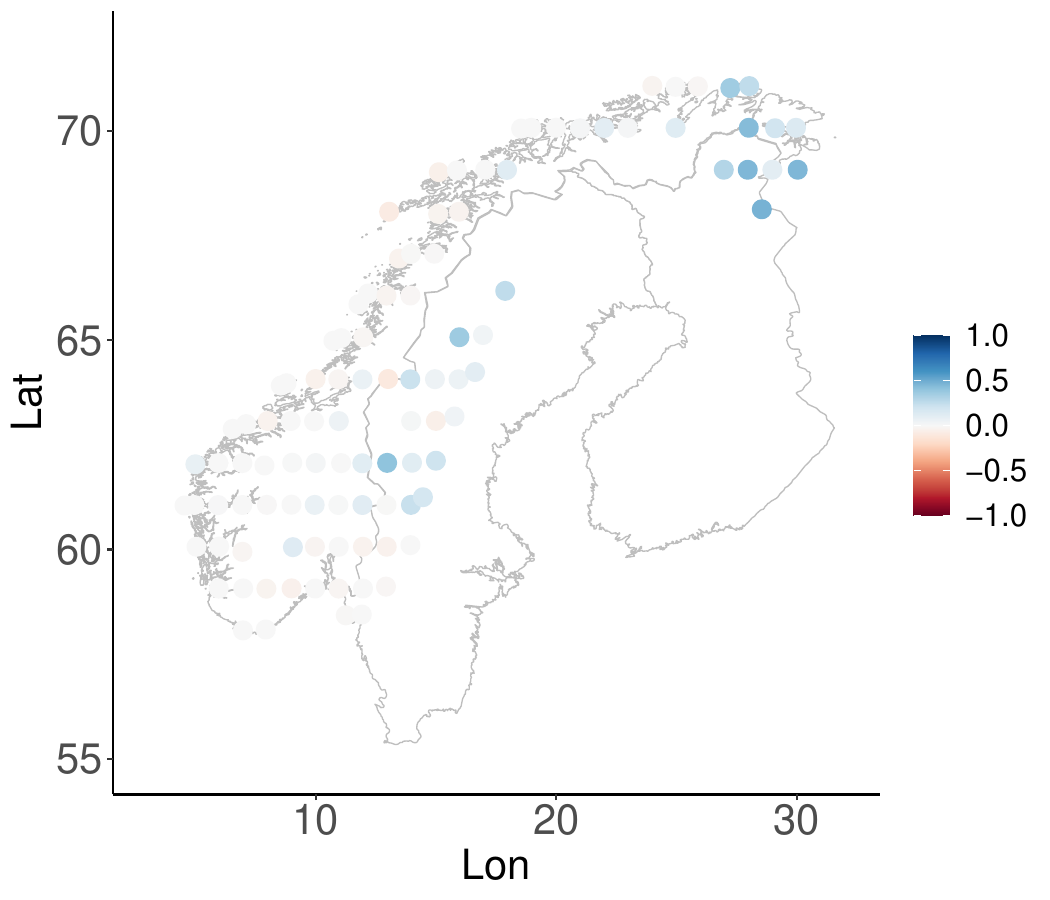}
         \caption{IBS skill score relative to the seasonal forecast (and climatology).}
         \label{fig:seasonalmap}
     \end{subfigure}\hspace{3mm}
     \begin{subfigure}[b]{0.30\textwidth}
         \centering
         \includegraphics[width=\textwidth]{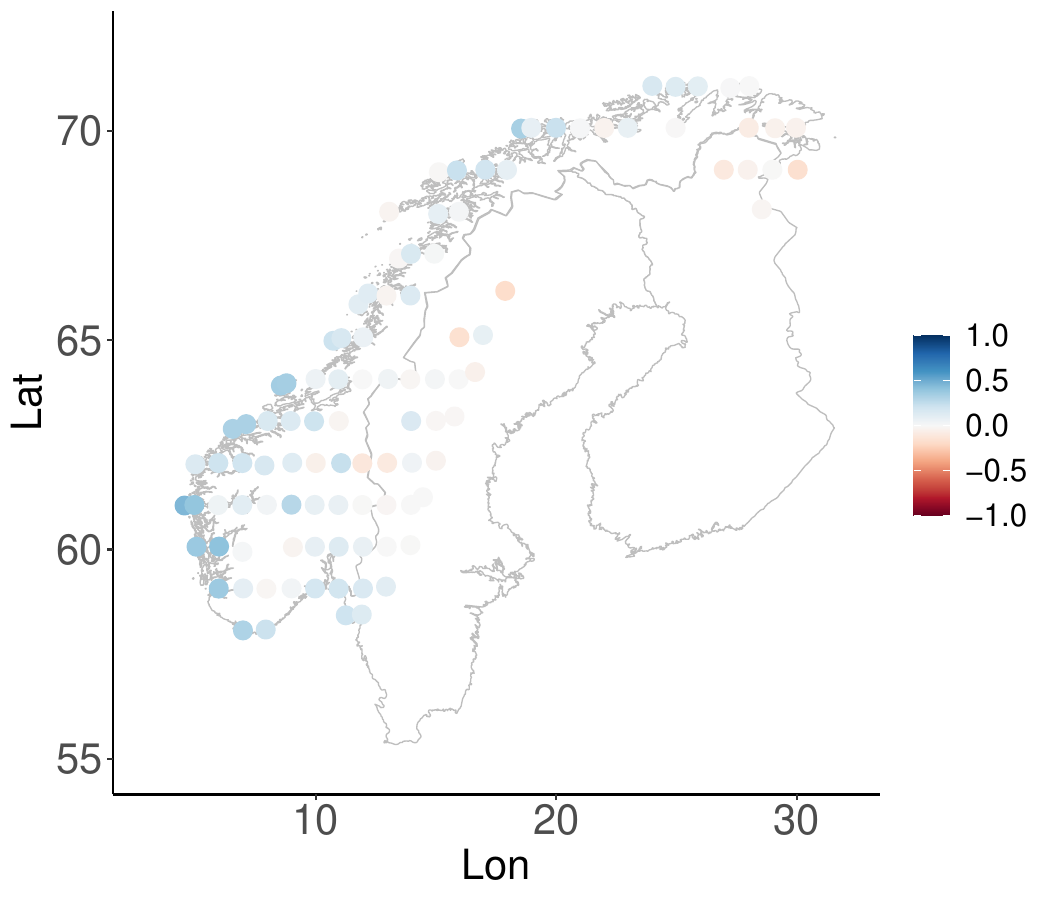}
         \caption{IBS skill score relative to the subseasonal forecast.}
         \label{fig:subseasonalmap}
     \end{subfigure}
          \begin{subfigure}[b]{0.30\textwidth}
         \centering
         \includegraphics[width=\textwidth,trim={0 0 0cm 0},clip]{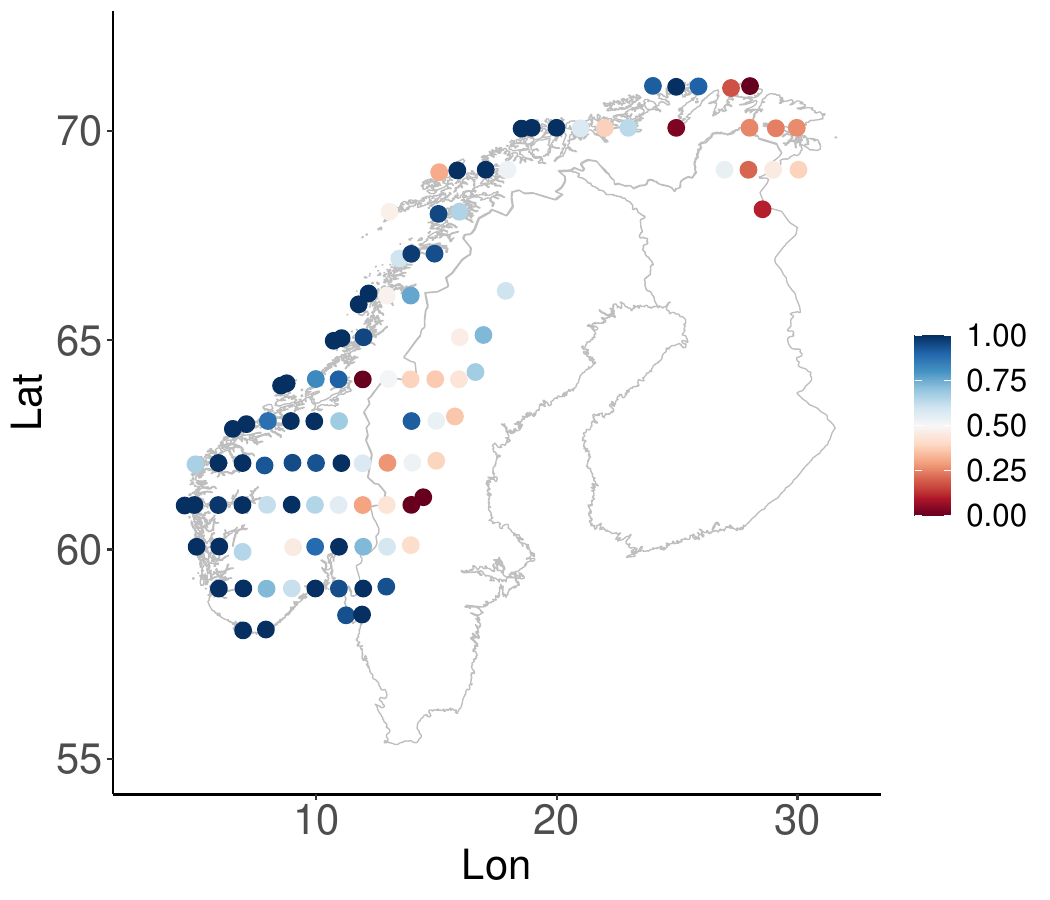}
         \caption{Weight $\omega$ on the seasonal forecast in the LP (ML) method.}
         \label{fig:weightmap}
     \end{subfigure}
        \caption{\small Location-specific skill scores of the LP (ML) combination forecast relative to the  seasonal (Figure \ref{fig:seasonalmap}) and the subseasonal forecast (Figure \ref{fig:subseasonalmap}) when post-processed forecasts are considered. As the seasonal forecast and the climatology forecast are similar in terms of IBS (see Figure \ref{fig:IBS_LP_realdata}), the skill scores of LP compared to the climatology forecast (Clim, KM) is almost identical to what is shown in Figure \ref{fig:seasonalmap}. Figure \ref{fig:weightmap} presents the median weight parameter $\omega$ for each location, and shows how much weight is put on the seasonal forecast in the the LP (ML) method. }
        \label{fig:resultmaps}
\end{figure}

Figure \ref{fig:resultmaps} shows how the skill scores of the LP (ML) method are distributed across our study locations relative to the subseasonal and seasonal forecasts. The parameter $\omega$ from Eq.~\eqref{eq:LPmethod}, that determines how much weight is assigned to the seasonal forecast in the LP method, is also displayed. We see that the combination method is considerably better than the subseasonal forecasts in many regions (Figure \ref{fig:subseasonalmap}), but particularly in southern Norway and along the coast. In coastal areas, hard freeze occurs late, as shown in Figure \ref{fig:mean_obs_ttf}, and often after the censoring date of the subseasonal forecast (day 76). It therefore makes sense that the seasonal forecast with a longer lead time is weighted more in these regions, as shown in Figure  \ref{fig:weightmap}.  

Compared to the seasonal forecast, the LP (ML) method is better in regions where frost comes early, i.e. in the northernmost parts of the study area and in the eastern parts, far away from the coast (Figure \ref{fig:seasonalmap}). These are regions where the subseasonal forecast performs well according to Figure \ref{fig:subseasonalmap}, and the combination method is able to incorporate this by giving the subseasonal forecast a high weight (red areas in Figure \ref{fig:weightmap}). In the remaining areas, the combination method has a similar performance as the seasonal forecast (and thus also the climatology).

\section{Discussion and concluding remarks}
\label{sec:conclusion}

In this paper, we have investigated methods for combining forecasts from different sources of information, in particular the combination of meteorological forecasts for time-to-event outcomes.
We have introduced two new combination methods, the Gaussian pool and hazard blending, and compared their performance with state-of-the art methods in an extensive simulation study and in a case-study where the goal is to predict the time until the first hard freeze of the season.

In the simulation study we found that combination forecasts typically outperform the single-source forecasts they are based on when there is sufficient amounts of training data \textit{or} when the sources are reasonably balanced. On the other hand, all combination methods struggle when there is little training data \textit{and} unbalanced sources. In these scenarios the more complex combination methods are generally worse than the simpler ones. The time-to-hard-freeze case-study provides both support and nuance to these findings. 
In the case-study we have a small training dataset (19 forecast-observation pairs) and we have two sources which are clearly quite unbalanced: the seasonal forecast has a higher predictive performance than the subseasonal forecast, at least when looking at the mean IBS over all locations. 
From the simulation study, we could therefore have expected that the case-study would be a challenging situation for the combination methods. 
Fortunately, we nonetheless obtained combination forecasts with a higher predictive power than the single-source forecasts, and which also outperformed the climatology forecast. Generally, the simpler single-parameter combination methods fared better than the more complex ones, and this was inline with our expectations from the simulation study.

We believe that these findings are likely to hold in the context of combination for other meteorological outcomes, beyond time-to-event set-ups which we have studied. 
Similar results have been reported elsewhere, for example in \citet{baran2018} where simpler combination methods generally outperformed the more complex ones in their case-studies. 
In the econometrics literature however, combination methods generally seem to struggle more than we have seen in our investigations. Many studies report that estimating combination weights from training data often results in a worst forecast than using fixed, equal weights, see for instance \citet{wang2022}. This phenomenon is referred to as the `forecast combination puzzle' and is discussed in for instance \citet{wang2022, zischke2022, frazier2023}. We do not observe this phenomenon: the simple benchmark method LP$_0$ with fixed, equal weights produces forecasts which have lower predictive performance than most other combination methods in both the simulation study and the case-study.
The reason might be, as mentioned in the introduction, that forecasts in meteorology are fundamentally different from forecasts in for instance economics. In particular, the forecasts in meteorology originate from physical models rather than statistical models on historical data, and this likely makes meteorology a field where combination methods should be particularly fruitful. 

In \citet{gneiting2013} the authors argue that combination methods ought to be \emph{flexibly dispersive}, so that they produce calibrated forecasts. To achieve this property, more combination parameters are required, and we see that in practice the fruitfulness of such approaches will depend on the size of the available training dataset and on aspects of the sources which we are combining. Still, we certainly think that the more complex, flexibly dispersive methods have a role to play, particularly in situations where the data sources are reasonably balanced or when there is sufficient training data. Also, we can consider more complex estimation schemes for the combination parameters, which could reduce the risk of overfitting to the limited training dataset. One may for example use regularisation methods in the estimation, as considered in  \citet{diebold2022}. And also more complex modelling of the combination parameters across for example time and space, as is often done in econometrics (see for instance  \citet{aastveit2022} and \citet{wang2022}). Specifically, for our application one could have considered to let some or all combination parameters depend on spatial covariates like for example elevation or spatial random effects (i.e. assuming some degree of similarity in the parameters that belong to locations that are close). This would have allowed us to borrow strength across locations and would likely have resulted in more stable estimation of the combination parameters. We plan to pursue ideas like these in future work. 

\section{Acknowledgements}
This work was supported by the Research Council of Norway through grant 309562589 (Climate Futures) and by the European Union's Horizon research and innovation program through grants 101081555 (Impetus4Change) and 869730 (CONFER). The authors thank Michael Scheuerer for preparing the raw, subseasonal forecast data and Thordis Thorarinsdottir for helpful discussion.
Computer simulations were performed in HPC solutions provided by the Oslo Centre for Epidemiology and Biostatistics, at the University of Oslo.

\bibliographystyle{biometrika}
\bibliography{merging.bib}

\section*{Appendix}
\setcounter{subsection}{0}
\renewcommand{\thesubsection}{\Alph{subsection}}

\setcounter{table}{0}
\renewcommand{\thetable}{A\arabic{table}}

\subsection{Flexible dispersivity of the Gaussian pooling}

Here we show that the Gaussian pooling is flexibly dispersive relative to the class $\mathcal C_\R$ of continuous probability distributions on $\R$.
For this result we adapt the setting of \citet{gneiting2013}. Our proof is largely analogous to the proof of their Proposition 3.5.

\begin{thm}
Let $Y$ have distribution $F_0 \in \mathcal C_\R$ and let $F_1,...,F_n \in \mathcal C_\R$. 
Denote by $G$ the combination method $G_{\mu,\sigma}:(F_1,...,F_n)\mapsto \Phi\bigg(\frac {\sum_{i=1}^n \omega_i \Phi^{-1}(F_i) - \mu}{\sigma}\bigg)$, and by $Z_{\mu,\sigma}$ the probability integral transform $G_{\mu,\sigma}(Y)$. As $(\mu,\sigma)$ varies over $\R\times \R_+$, the variance of $Z_{\mu,\sigma}$ attains any value in the open
interval (0,1/4). In particular, $G_{\mu,\sigma}$ is flexibly dispersive.
\label{th1}
\end{thm}
\begin{proof} The variance of $Z_{\mu,\sigma}$ is continuous in $\mu$ and $\sigma$. For $\sigma\to \infty$, the variance of $Z_{0,\sigma}$ converges to 0. Therefore, it is sufficient to find a $\mu_0$ and a sequence $(\sigma_i)_{i\in\N}$ such that $Var(Z_{\mu_0,\sigma_i})\to 1/4$. Let $y_0$ be a median of $F_0$, and let $\mu_0 := \sum_{i=1}^n \omega_i \Phi^{-1}(F_i(y_0))$. This implies $G_{\mu_0,\sigma_i}(y_0) = 1/2$ for any $\sigma_i>0$. Now, for any sequence $(\sigma_i)_{i\in\N}$ converging to 0, the law of $G_{\mu_0,\sigma_i}(Y)$ converges weakly to the Bernoulli measure with probability 1/2. This implies $Var(Z_{\mu_0,\sigma_i})\to 1/4$.
\end{proof}

\newpage

\subsection{Additional results from simulations}
\label{app:SimResults}

\begin{table}[ht] 
\caption{\small Mean of PIT values for 17 forecasting methods in 16 different scenarios. }\label{tab:appExp}
\centering \small
\begin{tabular}{rrrrrrrrrrrrrrrrr}
  \hline
 & 1 & 2 & 3 & 4 & 5 & 6 & 7 & 8 & 9 & 10 & 11 & 12 & 13 & 14 & 15 & 16 \\ 
  \hline
source 1 & 0.50 & 0.50 & 0.50 & 0.50 & 0.50 & 0.50 & 0.50 & 0.50 & 0.50 & 0.50 & 0.50 & 0.50 & 0.50 & 0.50 & 0.50 & 0.50 \\ 
  source 2 & 0.50 & 0.50 & 0.73 & 0.73 & 0.50 & 0.50 & 0.69 & 0.70 & 0.50 & 0.50 & 0.73 & 0.73 & 0.50 & 0.50 & 0.70 & 0.70 \\ 
  source 1 (KM) & 0.50 & 0.50 & 0.50 & 0.50 & 0.50 & 0.50 & 0.50 & 0.50 & 0.50 & 0.50 & 0.50 & 0.50 & 0.50 & 0.50 & 0.50 & 0.50 \\ 
  source 2 (KM) & 0.50 & 0.50 & 0.73 & 0.73 & 0.50 & 0.50 & 0.70 & 0.71 & 0.50 & 0.50 & 0.73 & 0.73 & 0.50 & 0.50 & 0.71 & 0.70 \\ 
  LP & 0.50 & 0.50 & 0.53 & 0.53 & 0.50 & 0.50 & 0.50 & 0.51 & 0.50 & 0.50 & 0.53 & 0.55 & 0.50 & 0.50 & 0.52 & 0.52 \\ 
  BP$_3$ & 0.51 & 0.50 & 0.51 & 0.50 & 0.51 & 0.50 & 0.50 & 0.50 & 0.50 & 0.50 & 0.50 & 0.50 & 0.50 & 0.50 & 0.50 & 0.50 \\
  GP$_3$ & 0.51 & 0.50 & 0.51 & 0.50 & 0.51 & 0.50 & 0.51 & 0.50 & 0.50 & 0.50 & 0.50 & 0.50 & 0.50 & 0.50 & 0.51 & 0.50 \\ 
  HB & 0.50 & 0.50 & 0.56 & 0.56 & 0.50 & 0.50 & 0.52 & 0.53 & 0.50 & 0.50 & 0.56 & 0.57 & 0.50 & 0.50 & 0.53 & 0.54 \\ 
  LP$_\ibs$ & 0.50 & 0.50 & 0.56 & 0.55 & 0.50 & 0.50 & 0.52 & 0.52 & 0.50 & 0.50 & 0.56 & 0.57 & 0.50 & 0.50 & 0.53 & 0.54 \\ 
  BP$_\ibs$ & 0.47 & 0.46 & 0.48 & 0.47 & 0.50 & 0.49 & 0.49 & 0.49 & 0.46 & 0.46 & 0.46 & 0.47 & 0.48 & 0.49 & 0.49 & 0.49 \\ 
  GP$_\ibs$ & 0.47 & 0.46 & 0.48 & 0.46 & 0.50 & 0.49 & 0.49 & 0.49 & 0.46 & 0.46 & 0.46 & 0.47 & 0.48 & 0.48 & 0.49 & 0.49 \\ 
  BP$_2$ & 0.50 & 0.50 & 0.57 & 0.56 & 0.50 & 0.50 & 0.51 & 0.52 & 0.50 & 0.50 & 0.56 & 0.57 & 0.50 & 0.50 & 0.53 & 0.53 \\
  GP$_1$ & 0.50 & 0.50 & 0.57 & 0.57 & 0.50 & 0.50 & 0.52 & 0.53 & 0.50 & 0.50 & 0.57 & 0.58 & 0.50 & 0.50 & 0.53 & 0.54 \\ 
  GP$_2$ & 0.50 & 0.50 & 0.58 & 0.57 & 0.50 & 0.50 & 0.52 & 0.53 & 0.50 & 0.50 & 0.57 & 0.58 & 0.50 & 0.50 & 0.54 & 0.54 \\ 
  GP$_3$-t & 0.51 & 0.50 & 0.51 & 0.50 & 0.51 & 0.50 & 0.51 & 0.50 & 0.50 & 0.50 & 0.50 & 0.50 & 0.50 & 0.50 & 0.50 & 0.50 \\ 
  LP$_0$ & 0.50 & 0.50 & 0.62 & 0.61 & 0.50 & 0.50 & 0.60 & 0.60 & 0.50 & 0.50 & 0.61 & 0.61 & 0.50 & 0.50 & 0.60 & 0.60 \\ 
  merge & 0.50 & 0.50 & 0.55 & 0.62 & 0.50 & 0.50 & 0.53 & 0.61 & 0.50 & 0.50 & 0.54 & 0.62 & 0.50 & 0.50 & 0.54 & 0.61 \\ 
   \hline
\end{tabular}

\bigskip

\caption{\small Standard deviation of PIT values for 17 forecasting methods in 16 different scenarios. }\label{tab:appSD}
 \small
 \centering
\begin{tabular}{rrrrrrrrrrrrrrrrr}
  \hline
 & 1 & 2 & 3 & 4 & 5 & 6 & 7 & 8 & 9 & 10 & 11 & 12 & 13 & 14 & 15 & 16 \\ 
  \hline
source 1 & 0.29 & 0.29 & 0.29 & 0.29 & 0.29 & 0.29 & 0.29 & 0.29 & 0.29 & 0.29 & 0.29 & 0.29 & 0.29 & 0.29 & 0.29 & 0.29 \\ 
  source 2 & 0.29 & 0.29 & 0.25 & 0.25 & 0.29 & 0.29 & 0.26 & 0.26 & 0.29 & 0.29 & 0.25 & 0.25 & 0.29 & 0.29 & 0.26 & 0.26 \\ 
  source 1 (KM) & 0.29 & 0.30 & 0.29 & 0.30 & 0.29 & 0.30 & 0.29 & 0.30 & 0.29 & 0.30 & 0.29 & 0.30 & 0.29 & 0.30 & 0.29 & 0.30 \\ 
  source 2 (KM) & 0.31 & 0.31 & 0.26 & 0.26 & 0.30 & 0.31 & 0.27 & 0.27 & 0.30 & 0.31 & 0.26 & 0.26 & 0.30 & 0.31 & 0.27 & 0.27 \\ 
  LP & 0.25 & 0.25 & 0.27 & 0.27 & 0.29 & 0.28 & 0.29 & 0.28 & 0.27 & 0.27 & 0.27 & 0.27 & 0.28 & 0.28 & 0.28 & 0.28 \\ 
  BP$_3$ & 0.29 & 0.28 & 0.29 & 0.28 & 0.29 & 0.29 & 0.29 & 0.29 & 0.31 & 0.31 & 0.31 & 0.31 & 0.31 & 0.31 & 0.31 & 0.31 \\ 
  GP$_3$ & 0.29 & 0.29 & 0.29 & 0.29 & 0.29 & 0.29 & 0.29 & 0.29 & 0.31 & 0.31 & 0.31 & 0.31 & 0.31 & 0.31 & 0.31 & 0.31 \\ 
  HB & 0.26 & 0.26 & 0.25 & 0.26 & 0.28 & 0.29 & 0.28 & 0.29 & 0.27 & 0.27 & 0.26 & 0.26 & 0.28 & 0.29 & 0.28 & 0.29 \\ 
  LP$_\ibs$ & 0.25 & 0.25 & 0.25 & 0.25 & 0.28 & 0.28 & 0.28 & 0.28 & 0.26 & 0.26 & 0.26 & 0.26 & 0.28 & 0.28 & 0.28 & 0.28 \\ 
  BP$_\ibs$ & 0.30 & 0.29 & 0.30 & 0.30 & 0.29 & 0.30 & 0.29 & 0.29 & 0.32 & 0.32 & 0.32 & 0.32 & 0.31 & 0.31 & 0.31 & 0.32 \\ 
  GP$_\ibs$ & 0.30 & 0.29 & 0.30 & 0.29 & 0.29 & 0.29 & 0.29 & 0.29 & 0.32 & 0.32 & 0.32 & 0.32 & 0.31 & 0.31 & 0.31 & 0.32 \\ 
  BP$_2$ & 0.29 & 0.28 & 0.28 & 0.28 & 0.29 & 0.29 & 0.29 & 0.29 & 0.30 & 0.30 & 0.29 & 0.29 & 0.30 & 0.30 & 0.30 & 0.30 \\ 
  GP$_1$ & 0.27 & 0.27 & 0.27 & 0.27 & 0.28 & 0.29 & 0.28 & 0.28 & 0.27 & 0.27 & 0.27 & 0.27 & 0.28 & 0.28 & 0.28 & 0.29 \\ 
  GP$_2$ & 0.29 & 0.29 & 0.28 & 0.28 & 0.29 & 0.29 & 0.29 & 0.29 & 0.30 & 0.30 & 0.29 & 0.29 & 0.30 & 0.30 & 0.30 & 0.30 \\ 
  GP$_3$-t & 0.29 & 0.29 & 0.29 & 0.29 & 0.29 & 0.29 & 0.29 & 0.29 & 0.30 & 0.30 & 0.30 & 0.30 & 0.29 & 0.30 & 0.29 & 0.30 \\ 
  LP$_0$ & 0.25 & 0.25 & 0.23 & 0.23 & 0.27 & 0.27 & 0.25 & 0.25 & 0.25 & 0.25 & 0.23 & 0.23 & 0.27 & 0.27 & 0.25 & 0.25 \\ 
  merge & 0.27 & 0.25 & 0.26 & 0.23 & 0.28 & 0.27 & 0.27 & 0.25 & 0.27 & 0.25 & 0.26 & 0.23 & 0.28 & 0.27 & 0.27 & 0.25 \\ 
   \hline
\end{tabular}

\end{table}

\begin{table}[ht] 
\caption{\small Mean Integrated Brier scores (IBS) for 17 forecasting methods in 16 different scenarios. The lowest scores in each scenario are bold-faced.  }
\centering \small
\begin{tabular}{rrrrrrrrr}
  \hline
 & 1 & 2 & 3 & 4 & 5 & 6 & 7 & 8 \\ 
  \hline
source 1 & 0.0778 & 0.0794 & 0.0775 & 0.0795 & 0.0771 & 0.0808 & 0.0778 & 0.0791 \\ 
  source 2 & 0.0805 & 0.0800 & 0.1016 & 0.1001 & 0.0919 & 0.0933 & 0.1080 & 0.1092 \\ 
  source 1 (KM) & 0.0781 & 0.0807 & 0.0778 & 0.0808 & 0.0773 & 0.0823 & 0.0781 & 0.0804 \\ 
  source 2 (KM) & 0.0875 & 0.0867 & 0.1033 & 0.1015 & 0.0973 & 0.0985 & 0.1102 & 0.1115 \\ 
  LP & 0.0702 & 0.0708 & 0.0755 & 0.0765 & 0.0766 & 0.0797 & 0.0776 & 0.0786 \\ 
  BP$_3$ & 0.0697 & 0.0695 & 0.0704 & 0.0711 & 0.0763 & 0.0792 & 0.0772 & 0.0778 \\ 
  GP$_3$ & 0.0692 & 0.0690 & 0.0697 & 0.0697 & 0.0762 & \textbf{0.0789} & 0.0770 & \textbf{0.0774} \\ 
  HB & 0.0707 & 0.0718 & 0.0752 & 0.0766 & 0.0764 & 0.0803 & 0.0777 & 0.0794 \\ 
  LP$_\ibs$ & 0.0702 & 0.0705 & 0.0749 & 0.0759 & 0.0763 & 0.0794 & 0.0774 & 0.0785 \\ 
  BP$_\ibs$ & 0.0689 & 0.0689 & 0.0696 & 0.0702 & 0.0761 & 0.0791 & 0.0770 & 0.0776 \\ 
  GP$_\ibs$ & \textbf{0.0683} & \textbf{0.0685} & \textbf{0.0687} & \textbf{0.0692} & \textbf{0.0760} & 0.0790 & \textbf{0.0769} & \textbf{0.0774} \\ 
  BP$_2$ & 0.0692 & 0.0694 & 0.0751 & 0.0760 & 0.0762 & 0.0792 & 0.0775 & 0.0785 \\ 
  GP$_1$ & 0.0688 & 0.0692 & 0.0740 & 0.0746 & 0.0761 & 0.0790 & 0.0773 & 0.0781 \\ 
  GP$_2$ & 0.0688 & 0.0690 & 0.0746 & 0.0749 & \textbf{0.0760} & \textbf{0.0789} & 0.0775 & 0.0781 \\ 
  GP$_3$-t & 0.0692 & 0.0690 & 0.0697 & 0.0697 & 0.0762 & \textbf{0.0789} & 0.0770 & \textbf{0.0774} \\ 
  LP$_0$ & 0.0702 & 0.0704 & 0.0779 & 0.0781 & 0.0786 & 0.0807 & 0.0836 & 0.0841 \\ 
  merge & 0.0730 & 0.0700 & 0.0749 & 0.0778 & 0.0762 & 0.0806 & 0.0779 & 0.0847 \\ 
   \hline
\end{tabular}

\vspace{0.2cm}

\begin{tabular}{rrrrrrrrr}
  \hline
 & 9 & 10 & 11 & 12 & 13 & 14 & 15 & 16 \\ 
  \hline
source 1 & 0.0777 & 0.0798 & 0.0773 & 0.0803 & 0.0776 & 0.0797 & \textbf{0.0780} & \textbf{0.0801} \\ 
  source 2 & 0.0788 & 0.0814 & 0.0989 & 0.1008 & 0.0937 & 0.0932 & 0.1113 & 0.1087 \\ 
  source 1 (KM) & 0.0780 & 0.0812 & 0.0776 & 0.0818 & 0.0779 & 0.0812 & 0.0783 & 0.0816 \\ 
  source 2 (KM) & 0.0854 & 0.0878 & 0.1005 & 0.1021 & 0.0991 & 0.0986 & 0.1135 & 0.1110 \\ 
  LP & 0.0728 & 0.0743 & 0.0765 & 0.0787 & 0.0780 & 0.0801 & 0.0783 & 0.0803 \\ 
  BP$_3$ & 0.0727 & 0.0749 & 0.0742 & 0.0759 & 0.0808 & 0.0826 & 0.0814 & 0.0830 \\ 
  GP$_3$ & 0.0717 & 0.0741 & 0.0724 & 0.0740 & 0.0807 & 0.0821 & 0.0814 & 0.0825 \\ 
  HB & 0.0724 & 0.0739 & 0.0764 & 0.0789 & 0.0785 & 0.0809 & 0.0792 & 0.0819 \\ 
  LP$_\ibs$ & 0.0714 & 0.0730 & 0.0760 & 0.0782 & 0.0782 & 0.0799 & 0.0789 & 0.0807 \\ 
  BP$_\ibs$ & 0.0730 & 0.0749 & 0.0743 & 0.0762 & 0.0815 & 0.0829 & 0.0824 & 0.0836 \\ 
  GP$_\ibs$ & 0.0723 & 0.0744 & 0.0731 & 0.0749 & 0.0815 & 0.0828 & 0.0824 & 0.0833 \\ 
  BP$_2$ & 0.0709 & 0.0726 & 0.0770 & 0.0791 & 0.0788 & 0.0806 & 0.0796 & 0.0816 \\ 
  GP$_1$ & 0.0695 & 0.0715 & 0.0747 & 0.0768 & 0.0776 & \textbf{0.0794} & 0.0787 & 0.0803 \\ 
  GP$_2$ & 0.0699 & 0.0720 & 0.0759 & 0.0780 & 0.0786 & 0.0802 & 0.0798 & 0.0814 \\ 
  GP$_3$-t & 0.0715 & 0.0738 & \textbf{0.0722} & \textbf{0.0738} & 0.0806 & 0.0820 & 0.0813 & 0.0825 \\ 
  LP$_0$ & \textbf{0.0693} & 0.0713 & 0.0769 & 0.0786 & 0.0797 & 0.0801 & 0.0849 & 0.0846 \\ 
  merge & 0.0725 & \textbf{0.0708} & 0.0745 & 0.0784 & \textbf{0.0769} & 0.0800 & 0.0785 & 0.0851 \\ 
   \hline
\end{tabular}
\label{tab:appIBS}
\end{table}

\newpage

\subsection{Accounting for estimation uncertainty in the source-specific probabilistic forecasts}
\label{app:Snew}

As described in Section~\ref{sec:survmeth} in the main text, we estimate the source-specific survival forecast $S_k(t)$ using the ensemble members from source $k$. If one uses a parametric model for the ensemble, one may simply plug-in the estimated parameters into the appropriate parametric forms and proceeded as if these estimates are in fact the true parameters. When the sample size is low this will result in probabilistic predictions which are not perfectly calibrated, even under the true model. The reason is that we ignore the uncertainty due to the estimation of the parameters. 

Luckily, there are ways to make corrected versions of $\widehat S_{k}(t)$ which take into account the variability from the estimation. This is particularly simple for the log-normal model. Say we have estimated $\xi$ and $\tau$ from $n_k$ observations, assuming a log-normal model. The corrected survival forecast for $T$ will be of the following form
\begin{align*}
    \widehat S_{k,\fix}(t) = G_{n_k-1}\left(\frac{\log(t)-\hat \xi}{\hat \tau \sqrt{1+1/n_k}}\right),
\end{align*}
where $G_{n_k-1}(\cdot)$ is the CDF of a student-t distribution with $n_k-1$ degrees of freedom.

\subsection{Additional details about the case-study}
\label{app:CaseStudy}

\subsubsection{Data sources}
\label{app:DataSources}

We use temperature data from three different sources: (1) seasonal forecasts initialized on September 1; (2)~subseasonal forecasts initialized on September 30; and (3) the actual time-to-hard-freeze observations which we obtain from the gridded, observation-based data product SeNorge2018 (version 18.12).

\begin{table}[h!]\caption{\small Information about the mean daily temperature data used in the analysis. The forecasts and observations cover the period from 1999 to 2018. The number of ensembles provided in the hindcasts varies in the study period, as indicated by the years shown in parenthesis. This also applies for the system number of the forecast. The range of the seasonal forecasts is 5-7 months into the future, but we use December 31 as our maximum date.}\label{tab:dataprod}\footnotesize
\begin{tabular}{lllllll}\hline
      Forecast/data type        &    From      & Spatial resol. & Ens. members & System/version number(s) & Init. date & Max. date \\
              \hline
   Seasonal fcst     & ECMWF      & $1^\circ \times     1^\circ$                &   25 (<2017), 51              &           5    & 01.09               & 31.12           \\
Seasonal  fcst        & CMCC     &   $1^\circ \times     1^\circ$                 &               40   &       35        & 01.09               & 31.12           \\
    Seasonal  fcst      &DWD     &   $1^\circ \times     1^\circ$                 & 30 (<2018), 50                &    21 ($\leq$2016), 2         & 01.09               & 31.12           \\
Seasonal fcst & Météo France      &   $1^\circ \times     1^\circ$                 &       25 (<2017), 51           &     8 ($\leq$2016), 6          & 01.09               & 31.12           \\
Seasonal fcst & UK Met Office     &  $1^\circ \times     1^\circ$                  &            7 (<2017), 2    &        601 ($\leq$2016), 12 (2017), 13 (2018)       & 01.09               & 31.12           \\
Subseasonal  fcst       & ECMWF  &  $1.5^\circ \times     1.5^\circ$                  &  11                 &       Cycle CY46R1 active in 2019     & 30.09                   & 15.11               \\
Observations   & SeNorge2018 &    1 km $\times$ 1 km                & $\times$         &     Version 18.12         & $\times$             & $\times$         \\
 \hline
\end{tabular}
\end{table}

Seasonal forecasts of 2-meter temperatures, initialized on September 1, were downloaded from the Copernicus Climate Data Store (CDS) \footnote{\url{ https://confluence.ecmwf.int/display/CKB/C3S+Seasonal+Forecasts\%3A+datasets+documentation}}. Ensemble forecasts were available from five weather centres: The European Centre for Medium-Range Weather Forecasts (ECMWF), the Euro-Mediterranean Center for Climate Change (CMCC), Deutscher Wetterdienst (DWD), Météo France and the UK Met Office. Forecasts from 1992-2021 with lead times 5-7 months and spatial resolution $1 ^\circ \times 1 ^\circ$, were downloaded with corresponding hindcasts for the 
 same historical period (1992-2021). The number of ensemble members varies among the forecasting systems and years, but each year there are approximately $150$ ensemble members in total for the five systems combined. See Table \ref{tab:dataprod} for details about the forecasting systems involved in the analysis and the number of ensemble members. 

The subseasonal forecasts we use, are the 2-meter temperature forecasts delivered by ECMWF, initialized on September 30 \footnote{\url{ https://confluence.ecmwf.int/display/S2S/ECMWF+model+description}}. We use a forecast produced in 2019, with forecast cycle number CY46R1. Hindcasts from the same cycle number for 1999-2018 were available with a spatial resolution of $1.5^\circ \times 1.5 ^\circ$, 11 ensemble members and a maximum lead time of 46 days (November 15).

The  temporal resolution of both the subseasonal and seasonal forecasts was 6 hours, but to match the SeNorge data, the data were converted to mean daily temperatures by computing averages of the 6 hour temperatures.

The SeNorge product is a dataset developed by the Norwegian Meteorological Institute and contains mean daily temperature data for the mainland of Norway, and parts of Sweden, Finland and Russia 
 \citep{lussana2018}. The product is based on 2-meter temperature observations from meteorological stations that are interpolated to a 1 km $\times$ 1  km grid by using the optimal interpolation approach \citep{Gandin,Kalnay}. The optimal interpolation ensures reasonable temperature changes with elevation. See \cite{lussana2018} for details. The SeNorge temperature products have been used in a variety of studies within different scientific fields, both as model input and for evaluations. See e.g. \cite{senorgeapplic2,senorgeapplic1,senorgeapplic4,senorgeapplic3}. Here, we use the SeNorge data as a substitute for local, observed surface temperatures, as in \cite{roksvaag2023}.

A number of grid nodes from SeNorge were used as target locations for predicting the time-to-hard-freeze. These were selected as follows: for each whole degree longitude and latitude covering the SeNorge product, we search for a nearby SeNorge location with mean elevation between 0.5 and 800 m a.s.l. The elevation restriction is added as the study is motivated by applications in agriculture and we want to avoid grid nodes that mainly cover sea, lakes or mountains. In addition, we require that any chosen location should not be further away from a whole longitude, latitude combination than 50 km. Among the possible candidates, one location for each whole longitude, latitude combination was picked at random, which resulted in 138 target locations. In the combination experiments that follows, we investigate whether we are able to predict the time-to-hard-freeze from October 1 for the years 1999-2018. For some of the 138 locations, the first hard freeze of the season often occurs before October 1. Locations where this happened for more than 5 out of the 20 study years, were removed from the analysis. This left us with the 103 target locations shown in Figure \ref{fig:mean_obs_ttf}.

\subsubsection{Post-processing}
\label{app:PostProc}
To ensure that the forecasted temperatures better match the temperatures observed at the target locations, the seasonal and subseasonal forecasts are post-processed in the same manner as in \cite{roksvaag2023}, by adjusting the forecast means and standard deviations relative to historical means and standard deviations from SeNorge. More specifically, let $f_{s,t,y}^{k,m}$ be the forecasted mean daily temperature for year $y$, location $s$, forecasting system $k$, ensemble member $m$ and day number $t$ after September 1. We here use the forecast from the grid node that is closest to the SeNorge location we are interested in. The forecast  $f_{s,t,y}^{k,m}$ is first standardized with respect to its own historical mean temperature $\tilde{\mu}^k_{s,t}$ and standard deviation $\tilde{\sigma}^k_{s,t}$, i.e.
\begin{equation}\label{eq:anomaly}
	\tilde{f}_{s,t,y}^{k,m}=\frac{(f_{s,t,y}^{k,m}-\tilde{\mu}^k_{s,t})}{\tilde{\sigma}^k_{s,t}},
\end{equation}
where $\tilde{\mu}^k_{s,t}$ and $\tilde{\sigma}^k_{s,t}$ are computed based on temperature hindcasts for system $k$, day $t$ and location $s$, with the target year $y$ being left out of the sample. There are 20 hindcast years for the subseasonal forecast (1999-2018) and 30 for the seasonal (1992-2021) for computing means and standard deviations. Next, we compute the historical mean temperature $\hat{\mu}_{s,t}$ and standard deviation $\hat{\sigma}_{s,t}$ for day $t$  for our target SeNorge location $s$ based on the same historical period, and re-standardize the forecast in Eq.~\eqref{eq:anomaly} as follows:
\begin{equation}\label{eq:SFE}
    \hat{f}_{s,t,y}^{k,m}=\tilde{f}_{s,t,y}^{k,m} \cdot \hat{\sigma}_{s,t}  + \hat{\mu}_{s,t}.
\end{equation}
The post-processed  temperature forecasts $\hat{f}_{s,t,y}^{k,m}$ now have a similar mean and standard deviation to what has been historically observed for the SeNorge location in question. Evaluations in \cite{roksvaag2023} show that the above post-processing method both gives temperature forecasts, and time-to-hard-freeze forecasts, that are well-calibrated.

From the post-processed seasonal and subseasonal forecast datasets, we extract the day of the first occurrence of hard freeze after September 1 for each combination of year, location, ensemble member and forecast system. Likewise, we take out the first day of observed hard freeze for SeNorge. This is done for each year from 1999 to 2018, as these are the years covered by all three data products (additional years from 1992-2021 are only used for post-processing of the seasonal forecasts). We select December 31 as our last possible hard freeze occurrence date. The seasonal forecasts and the SeNorge data are therefore truncated on day 122 (December 31). The subseasonal forecasts, however, have a maximum lead time of 46 days (after September 30) and are therefore censored already on November 15. 

The properties of the three different data products are summarized in Table \ref{tab:dataprod}. In Figure \ref{fig:mean_obs_ttf} the study locations are shown, together with histograms of the observed and/or forecasted time-to-hard-freeze for the three datasets.

\subsubsection{Additional results}
\label{app:Results}

\begin{figure}[h] 
\begin{center}
\includegraphics[scale=0.73]{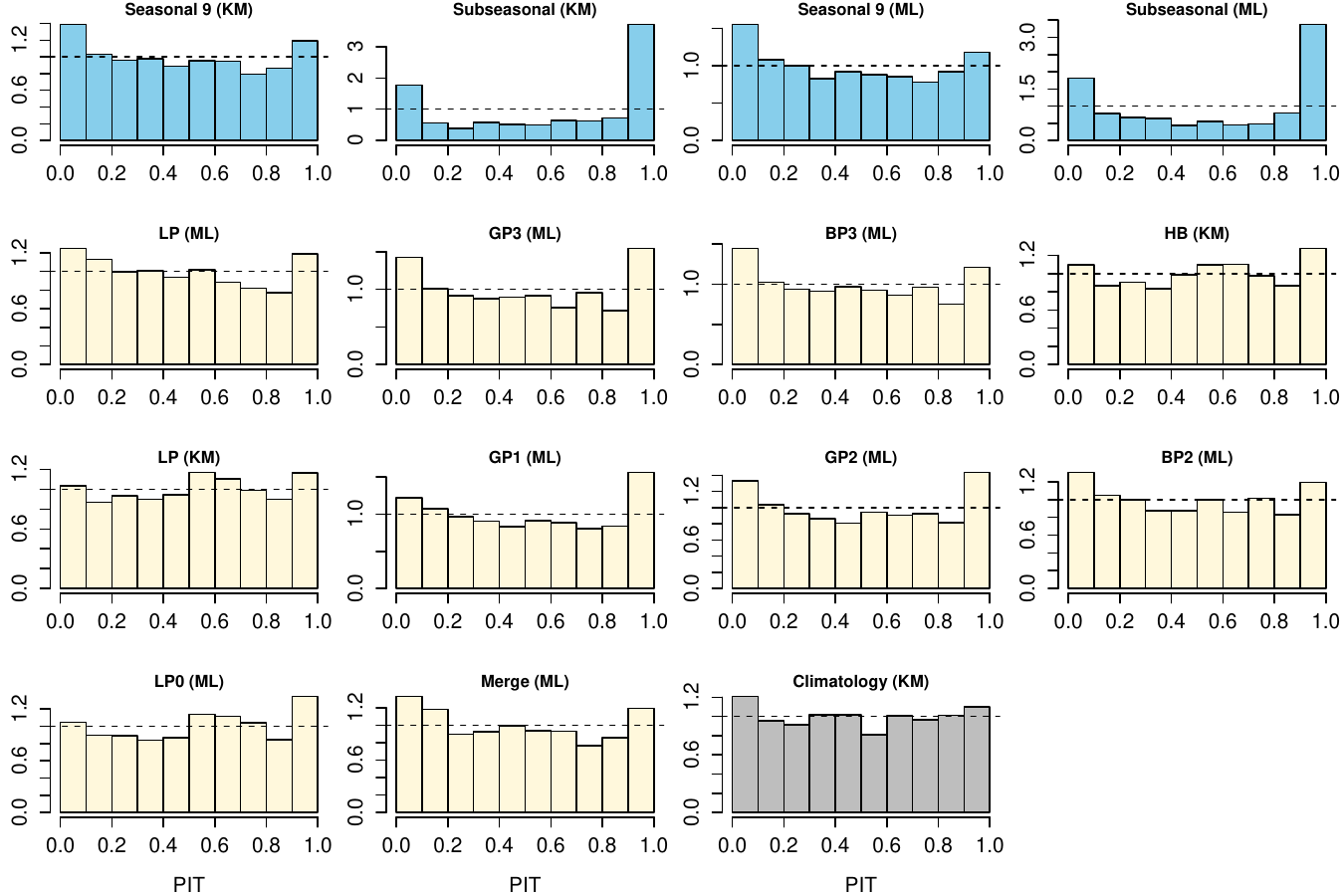}\\
\end{center} 
\caption{\small PIT histograms for the single-source forecasts (blue), 
the combination methods (yellow) and the climatology forecast (gray) for predictions of time-to-hard-freeze when raw forecasts are used.}
\label{fig:PIThist_raw}
\end{figure}

Figure \ref{fig:PIThist_raw} shows PIT histograms for the forecasts of the time hard freeze when raw ensembles are used as input. PIT histograms for the post-processed forecasts are already presented in Figure \ref{fig:pit_realdata}. Table \ref{tab:PITreal} shows the corresponding mean and standard deviations of the PIT value histograms for both raw and post-processed forecasts.

\begin{figure}[h] 
\begin{center}
\includegraphics[scale=0.8]{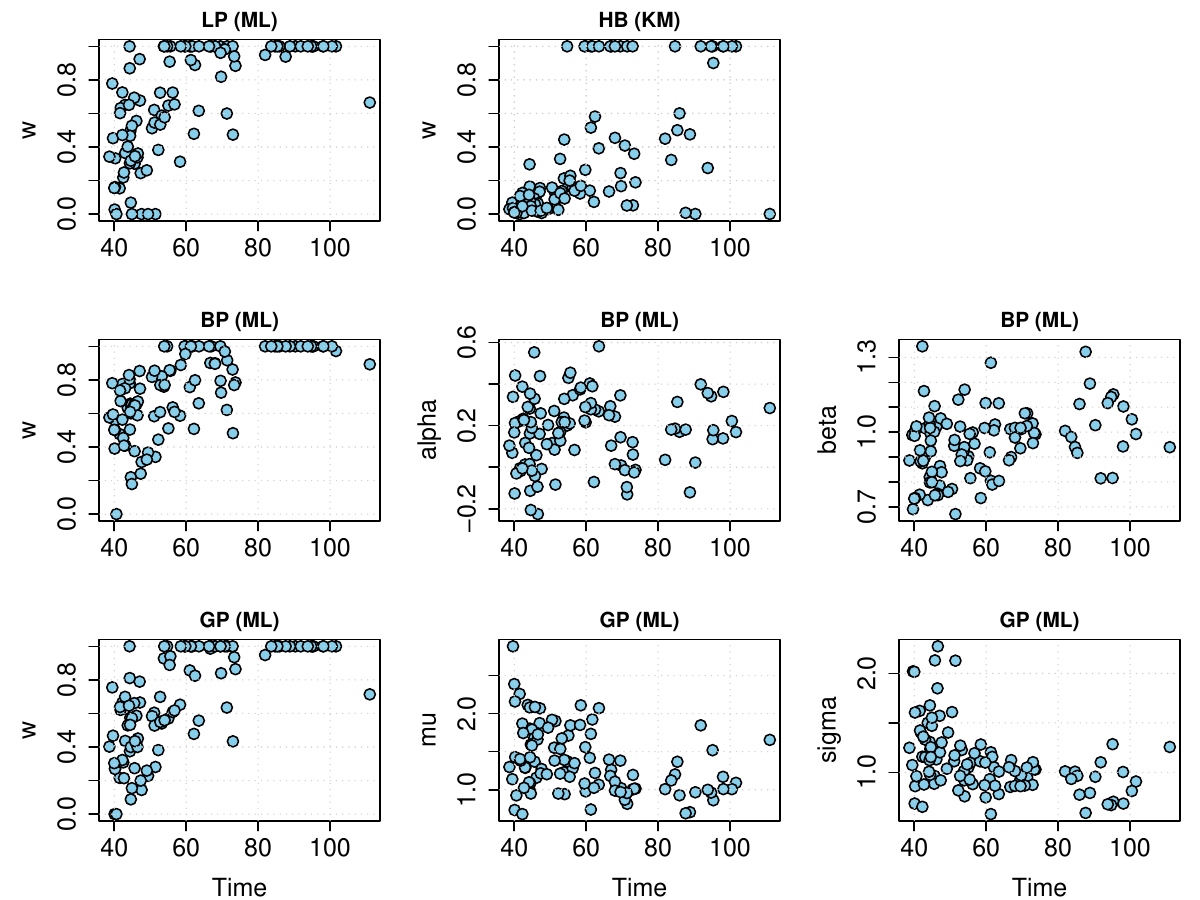}\\
\end{center} 
\caption{\small Parameter values for each main combination method, LP (ML), HB (KM), BP$_3$ (ML) and GP$_3$ (ML) when using post-processed time-to-hard-freeze forecasts. The median parameter value across the 20 study years is displayed for each location and shown as a function of the mean observed time-to-hard-freeze for the same years. }
\label{fig:parameter_values_real_pp}
\end{figure}

In Figure \ref{fig:parameter_values_real_pp}, the parameter values for the combination methods are shown for each location, as a function of the historically observed time-to-hard-freeze, for the post-processed forecasts. Figure \ref{fig:parameter_values_real_raw} shows the same, but for the raw forecasts.

\begin{table}[]\centering \caption{Mean of PIT values, with standard deviation of PIT values in parenthesis, for the time-to-hard-freeze case-study across 103 study locations and 20 years. Results for both post-processed and raw forecasts are included. For a calibrated forecast, the mean should be close to 0.5, with a standard deviation around 0.29.}\label{tab:PITreal}
\begin{tabular}{lll}\hline
Method       & Post-processed & Raw     \\
\hline
Clim (KM)    & 0.50 (0.31) & 0.50 (0.32) \\
Seas9 (KM)   & 0.53 (0.30) & 0.48 (0.26) \\
Subseas (KM) & 0.58 (0.32) & 0.61 (0.27) \\
Seas 9 (ML)  & 0.52 (0.30) & 0.47 (0.26) \\
Subseas (ML) & 0.53 (0.32) & 0.57 (0.27) \\
LP (ML)      & 0.53 (0.29) & 0.48 (0.27) \\
GP$_3$ (ML)  & 0.48 (0.31) & 0.51 (0.32) \\
BP$_3$ (ML)  & 0.48 (0.31) & 0.48 (0.31) \\
HB (KM)      & 0.55 (0.29) & 0.52 (0.27) \\
LP (KM)      & 0.55 (0.29) & 0.51 (0.26) \\
GP$_1$ (ML)  & 0.53 (0.30) & 0.50 (0.29) \\
GP$_2$ (ML)  & 0.53 (0.30) & 0.50 (0.29) \\
BP$_2$ (ML)  & 0.53 (0.30) & 0.49 (0.28)\\ 
LP$_0$ (ML) & 0.53 (0.28) & 0.52 (0.30) \\
Merge (ML) & 0.52 (0.29) &  0.48 (0.30)   \\

\hline
\end{tabular}
\end{table}

\begin{figure}[h] 
\begin{center}
\includegraphics[scale=0.8]{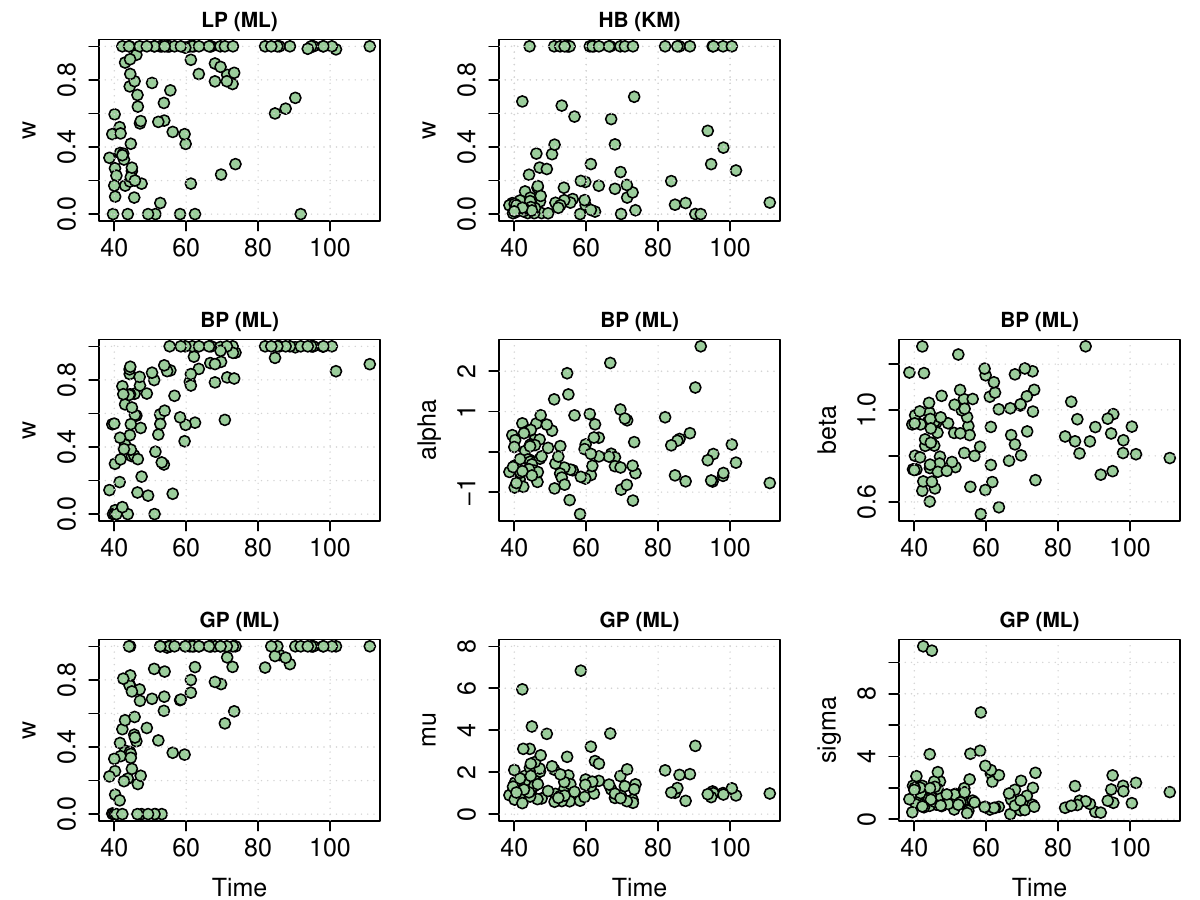}\\
\end{center} 
\caption{\small Parameter values for each main combination method, LP (ML), HB (KM), BP$_3$ (ML) and GP$_3$ (ML) when using raw time-to-hard-freeze forecasts. The median parameter value across the 20 study years is displayed for each location and shown as a function of the mean observed time-to-hard-freeze for the same years. For GP (ML) there was also one mu-value that around 25 that is not included in the picture.}
\label{fig:parameter_values_real_raw}
\end{figure}

\end{document}